\documentclass{article}

\usepackage{arxiv}

\usepackage{microtype}
\usepackage{graphicx}
\usepackage{subfigure}
\usepackage{booktabs} %

\usepackage{hyperref}

\usepackage{macros}

\title{Online Learning under Budget and ROI Constraints via Weak Adaptivity}

\author{%
  Matteo Castiglioni\thanks{Equal contribution. Correspondence to: <andrea.celli2@unibocconi.it>.}\\
  Politecnico di Milano\\
  \And
  Andrea Celli$^{\ast}$\\
  Bocconi University \\
  \And
  Christian Kroer\\
  Columbia University\\
}

\begin{document}

\maketitle

\begin{abstract}
    We study online learning problems in which a decision maker has to make a sequence of costly decisions, with the goal of maximizing their expected reward while adhering to budget and  return-on-investment (ROI) constraints. 
    Existing primal-dual algorithms designed for constrained online learning problems under adversarial inputs rely on two fundamental assumptions.
    First, the decision maker must know beforehand the value of parameters related to the degree of strict feasibility of the problem (i.e. Slater parameters). Second, a strictly feasible solution to the offline optimization problem must exist at each round. Both requirements are unrealistic for practical applications such as bidding in online ad auctions. 
    In this paper, we show how such assumptions can be circumvented by endowing standard primal-dual templates with \emph{weakly adaptive} regret minimizers. This results in a ``dual-balancing'' framework which ensures that dual variables stay sufficiently small, even in the absence of knowledge about Slater's parameter.
    We prove the first \emph{best-of-both-worlds} no-regret guarantees which hold in absence of the two aforementioned assumptions, under stochastic and adversarial inputs.
    Finally, we show how to instantiate the framework to optimally bid in various mechanisms of practical relevance, such as first- and second-price auctions.
\end{abstract}

\section{Introduction}

A decision maker takes decisions over $T$ rounds. At each round $t$, the decision $x_t\in\cX$ is chosen before observing a reward function $f_t$ together with a set of \emph{time-varying} constraint functions. 
The decision maker is allowed to make decisions that are \emph{not} feasible, provided that the overall sequence of decisions obeys the \emph{long-term constraints} over the entire time horizon, up to a small cumulative violation across the $T$ rounds. 
The goal of the decision maker is to maximize their cumulative reward, while satisfying the long-term constraints. This model was first proposed by \citet{mannor2009online} and later developed along various directions \citep{mahdavi2012trading,jenatton2016adaptive,liakopoulos2019cautious,yu2017online,castiglioni2022unifying}.

Motivated by applications in online ad auctions, we consider the case in which the decision maker has a budget and \emph{return-on-investment} (ROI) constraint \citep{auerbach2008empirical,golrezaei2021bidding,golrezaei2021auction}. The decision maker is subject to bandit feedback: at each time $t$, the decision maker takes a decision $x_t$ and then observes the realized reward $f_t(x_t)$ and a cost $c_t(x_t)$. Inputs $(f_t,c_t)$ can either be generated i.i.d. according to some unknown distribution, or selected by an oblivious adversary. 

A key challenge of our model is that ROI constraints are not \emph{packing}, thereby preventing the direct application of known algorithms for \emph{adversarial bandits with knapsacks} (ABwK) \citep{immorlica2019adversarial,Castiglioni2022Online}, or for online allocation problems with resource-consumption constraints \citep{balseiro2022best}. Moreover, previous work addressing the adversarial case with non-packing constraints makes the assumption that the ``worst-case feasibility'' with respect to all constraint functions observed up to $T$ is strictly positive \citep{sun2017safety,castiglioni2022unifying,immorlica2019adversarial,balseiro2022best}. In other words, there has to exist a ``safe'' policy guaranteeing that, at each $t$, the constraints can be satisfied by a margin at least $\alpha>0$, which has to be known in advance by the learner.
This can be problematic for at least two reasons: \textbf{i) $\alpha$ may not be known in advance to the decision maker, and ii) the safe policy may not exist in all rounds $t$.}
For example, in the case of bidding in repeated ad auctions under budget and ROI constraints, such assumptions do not hold. In particular, the decision maker would be required to have an action yielding expected ROI  strictly above their target for each round $t$. If we assume one ad placement is being allocated at each $t$ then this assumption is equivalent to assuming that the bidders' value is always strictly higher than the highest competing bid, which clearly would not hold in practice.
In this paper, we propose a general approach to enhance primal-dual frameworks based on the template by \citet{immorlica2019adversarial}, in order to obtain \emph{best-of-both-worlds} no-regret guarantees while bypassing both assumptions.

\subsection{Contributions} 

Previous primal-dual templates for adversarial inputs assume that the parameter $\alpha>0$ is known (see, \eg \cite{immorlica2019adversarial,balseiro2022best}). This is crucial to ensure boundedness of Lagrange multipliers, which is typically obtained by requiring the $\ell_1$-norm of the multipliers to be less than or equal to $1/\alpha$ \citep{Castiglioni2022Online,nedic2009approximate}.
The key contribution of the paper is showing that, in absence of any information on $\alpha>0$, if we require the primal and dual regret minimizers to be \emph{weakly adaptive} (\ie to guarantee sublinear regret on any interval $[t_1,t_2]\subseteq [T]$ \citep{hazan2007adaptive}), then boundedness of multipliers automatically emerges as a byproduct of the interaction between the primal and dual algorithms (thereby \textbf{solving challenge (i)}). 
Moreover, our framework only requires the existence of a safe policy ``frequently enough'', and not at all time steps $t$ (thereby \textbf{solving challenge (ii)}). 
This is the first time that the notion of adaptive regret minimization is used within primal-dual frameworks. Interestingly, this usage is significantly different from its original motivating applications \cite{hazan2007adaptive,adamskiy2016closer,daniely2015strongly}.

We show that the resulting framework provides best-of-both-worlds no-regret guarantees while solving both limitations. 
We prove a tight $\tilde O(T^{1/2})$ regret upper bound in the stochastic setting, and an $\alpha/(\alpha+1)$ constant-factor competitive ratio in the adversarial setting, under the standard assumption that the budget is $\Omega(T)$ \citep{balseiro2022best}. In both settings, our framework guarantees vanishing cumulative ROI constraint violation, and cumulative expenditure less than or equal to the available budget. 
Best-of-both-worlds algorithms for problems with long-term constraints typically require different proof techniques for the two input models. We unify most of the analysis, with the only difference being in the characterization of a particular set of policies.

Finally, we show how our framework can be employed for bidding in any mechanism with finite types. In particular, we show that it can handle the case of repeated \emph{non-truthful} auctions (\eg first-price auctions). Previous work could only handle budget- and ROI-constrained bidders in the simpler case of second-price auctions, in which truthfulness can be exploited \citep{feng2022online,golrezaei2021bidding}.

\subsection{Related works} 

Standard primal-dual approaches for bandit problems with knapsack constraints cannot be applied in our setting, as they require as an input the Slater's parameter $\alpha$ \citep{balseiro2019learning,immorlica2019adversarial,balseiro2022best,Castiglioni2022Online,castiglioni2022unifying}. In these works, the knowledge of $\alpha$ is exploited to ensure that dual variables are ``small'' through an explicit projection step over a set which depends on $\alpha$. This is not possible in our setting, due to the presence of non-packing ROI constraints.
In our new analysis, we show how such frameworks can be suitably adapted to work in more complex scenarios than the standard one. 
The issue of not knowing $\alpha$ has been effectively addressed in stochastic settings \citep{yu2017online,wei2020online,castiglioni2022unifying,lobos2021joint}. Nonetheless, since our goal is to provide guarantees that hold also in the presence of adversarial inputs, such results do not extend to our setting. 

\xhdr{Repeated auctions.} The problem of online bidding under budget constraints has been studied in various settings \citep{balseiro2019learning,ai2022no}. In the context of online allocation problems with an arbitrary number of constraints, \citet{balseiro2020dual,balseiro2022best} propose a class of primal-dual algorithms attaining asymptotically optimal performance in the stochastic and adversarial case. In their setting, at each round, the input $(f_t,c_t)$ is observed by the learner \emph{before} they make a decision. This makes the problem substantially different from ours. In particular, their framework cannot handle \emph{non-truthful} repeated auctions such as first-price auctions.
Recent works have also examined settings similar to ours, involving bidders with constraints on their budget and ROI.
The framework by \citet{feng2022online} can handle both ROI and budget constraints, but crucially relies on truthfulness of second-price auctions, and on the stochasticity of the environment. 
The recent work by \citet{wang2023learning} considers the problem of bidding in repeated first-price auctions under only budget constraints and stationary competition. Their analysis cannot be extended to our setting for the same reasons mentioned at the beginning of the section.

\xhdr{Concurrent work.} 
\citet{slivkins2023contextual} studies a stochastic setting with general long-term constraints similar to \citet{castiglioni2022unifying}. They provide a $\tilde O(T^{1/2})$ guarantee when $\alpha$ is known, and $\tilde O(T^{3/4})$ guarantees when $\alpha$ is not known. 
The latter result cannot be extended to the case of inputs generated by an  adversary.
The recent paper by~\citet{bernasconi2023bandits} studies a different setting from ours, in which they only have budget constraints, and knowledge of $\alpha$ is hindered by the fact that resources can be replenished (\ie costs can be negative, as in \citet{kumar2022non}). They provide best-of-both-world guarantees by exploiting results presented in this paper. 
Further, related works are presented in \cref{sec:apprelated}.

\section{Preliminaries}\label{sec:prel}

At each round $t\in[T]$, the decision maker chooses an action $x_t\in \cX$, where $ \cX$ is non-empty set of available actions, and subsequently observes reward $f_t(x_t)$ with $f_t:\cX\to[0,1]$, and incurs a cost $c_t(x_t)$, with $c_t:\cX\to[0,1]$.
We denote as $\cF$, respectively $\cC$, the set of all the possible functions $f_t$, respectively $c_t$ (\eg $\cF$ and $\cC$ may contain all the Lipschitz-continuous functions defined over $\cX$, or all the convex functions over $\cX$). We assume that functions in $\cF$ and $\cC$ are measurable with respect to probability measures over $\cX$. This ensures that expectations are well-defined, since the functions are assumed to be bounded above, and they are therefore integrable.
Following previous work \citep{agarwal2014budget,Badanidiyuru2018jacm,immorlica2019adversarial}, we assume the existence of a \emph{void action} $\nullx$ such that, for any pair $(f,c)\in\cF\times\cC$, $f(\nullx)=c(\nullx)=0$.
The decision maker has an overall budget $B \in \mathbb{R}_+$, $B=\Omega(T)$, which limits the total expenditure throughout the $T$ rounds. We denote by $\rho>0$ the \emph{per-iteration budget} defined as $B/T$.
Moreover, the decision maker has a target \emph{return-on-investments} (ROI) $\roi>0$. In order to simplify the notation, throughout the paper we will assume $\roi\defeq 1$. This comes without loss of generality: whenever $\roi>1$ we can suitably scale down values of reward functions $f_t$.
Then, the decision maker has the goal of maximizing their cumulative utility $\sum_{t=1}^T f_t(x_t)$, subject to the following constraints:
\begin{OneLiners}
    \item \textbf{Budget constraints}: $\sum_{t=1}^T c_t(x_t)\le \rho T$. Such constraints should be satisfied ``no matter what,'' so we refer to them as \emph{hard} constraints.
    \item \textbf{ROI constraints}: $\sum_{t=1}^T \mleft( c_t(x_t) - f_t(x_t)  \mright) \le 0.$
    We say ROI constraints are \emph{soft} meaning that we allow, in expectation, a small (vanishing in the limit) cumulative violation across the $T$ rounds.
\end{OneLiners}

In the context of repeated ad auctions, as we will discuss in \cref{sec:app}, this model can be easily instantiated to describe \emph{any} mechanism with finite types beyond the well-studied case of second-price auctions.

\xhdr{Auxiliary LP.} We endow $\cX$ with the Lebesgue $\sigma$-algebra, and we denote by $\Pi$ be the set of \emph{randomized policies}, defined as the set of probability measures on the Borel sets of $\cX$. 
At any $t\in [T]$ the decision maker will compute a policy $\vpi_t\in\Pi$ and play an action $x_t\sim\vpi_t$ accordingly.\footnote{The set $\{1,\ldots,n\}$, with $n\in\N$, is compactly denoted as $[n]$, and we let $[0]$ be equal to the empty set. Moreover, given a discrete set $\cX$, we denote by $\Delta_\cX$ the $|\cX|$-simplex.}
Given a reward function $f$ and a cost function $c$, let $g:\Pi\ni\vpi\mapsto\E_{x\sim\vpi}\mleft[ c(x)\mright]-\rho$ be the expected gap between the cost for policy $\vpi$ and the per-iteration budget $\rho$, and $h:\Pi\ni\vpi\mapsto \E_{x\sim\vpi}\mleft[c(x)-f(x)\mright]$ be the expected ROI constraint violation for policy $\vpi$. We will denote by $g_t$, resp. $h_t$, the constraints defined for the pair $(f_t,c_t)$ observed at round $t$.
In order to simplify the notation, given $x\in\cX$, the value of the reward function for the policy that derministically plays action $x$ (\ie the Dirac mass $\delta_x$) will be denoted by $f_t(x)$ in place 
of $f_t(\delta_x)$. Analogously, we will write $c_t(x)$, $g_t(x)$, and $h_t(x)$ instead of using Dirac measures $\delta_x$. 
Let $\distr$ be an arbitrary probability measure over the space of possible inputs $\cF\times\cC$. Then, we define the linear program \ref{eq:opt lp gen} as follows:
\begin{equation}\label{eq:opt lp gen}
	\tag{$\LP_{\distr}$}
	\OPT_{\distr}\defeq\mleft\{\begin{array}{lll}
		\displaystyle
		\sup_{\vpi \in \Pi} & \E_{f\sim\distr}f(\vpi)  \\
		\,\,\text{\normalfont s.t. }& \E_{\distr} g(\vpi)\le 0\\
		& \E_{\distr} h(\vpi)\le 0
	\end{array}\mright..
\end{equation}
\ref{eq:opt lp gen} selects the bidding policy $\vpi$ maximizing the expected reward according to $\distr$, while ensuring that constraints $g$ and $h$ encoded by $\distr$ are satisfied in expectation (both $g$ and $h$ are defined by $(f,c)\sim\distr$).
The \emph{Lagrangian function} $\cL_{\distr}:\Pi \times\R_{\ge 0}^2 \to\R$ of the above LP is defined as 
\[
	\cL_{\distr}(\vpi,\lambda,\mu)\defeq \E_{(f,c) \sim \distr} \mleft[f(\vpi) - \lambda g(\vpi)-\mu h(\vpi)\mright].
\]

\section{Baselines}\label{sec:baselines}

Our goal is to design online algorithms that output a sequence of policies $\vpi_1,\ldots,\vpi_T$ such that i) the \emph{cumulative regret} with respect to the performance of the baseline grows sublinearly in $T$, ii) the budget constraint is (deterministically) satisfied, \ie $\sum_{t=1}^Tc_t(x_t)\le B$, and iii) the \emph{cumulative ROI constraint violation} $\sum_{t=1}^T h_{t}(\vpi_t)$ grows sublinearly in the number of rounds $T$.
The cumulative regret of the algorithm is defined as $R^T \coloneqq T \, \OPT - \sum_{t=1}^T f_t(x_t)$,
where the baseline $\OPT$ depends on how the input sequence $\gamma\defeq(f_t,c_t)_{t=1}^T$ is generated. 
We consider the following two settings for which we define an appropriate value of the baseline, and a suitable problem-specific parameter $\alpha\in\R$ which is related to the feasibility of the offline problem.

\xhdr{Stochastic setting}: at each $t\in[T]$, the pair $(f_t,c_t)$ is independently drawn according to a fixed but unknown distribution $\distr$ over $\cF\times\cC$. The baseline is $\OPT_\distr$, which is the standard baseline for stochastic BwK problems since its value is guaranteed to be closed to that of the best dynamic policy \citep[Lemma 3.1]{Badanidiyuru2018jacm}.
 In this setting, let
	$
		\alpha \defeq -\inf_{\vpi \in \Pi}\max\mleft\{ \E_\distr g(\vpi), \E_\distr h(\vpi) \mright\}.
	$
	
	\xhdr{Adversarial setting}: the sequence of inputs $\gamma$ is selected by an oblivious adversary. Given $\gamma$, we define the following distribution over inputs: for any pair $(f,c)\in\cF\times\cC$, $\bar\gamma[f,c]=\sum_{t=1}^T\indicator{f_t=f, c_t=c}/T$. Then, the baseline is the solution of $\LP_{\bar\gamma}$ (\ie $\OPT_{\bar\gamma}$), which is the standard baseline for the adversarial setting (see, \eg \citet{balseiro2022best,immorlica2019adversarial}). Therefore, the baseline is obtained by solving the offline problem initialized with the average of the realizations observed over the $T$ rounds. Moreover, our results will also hold with respect to the best \emph{unconstrained} policy.
	We define $\alpha$ as 
	$\alpha \defeq -\inf_{\vpi} \max_{t \in [T]}\max\mleft\{  g_t(\vpi), h_t(\vpi) \mright\}.$
	In this setting, $\alpha$ represents the ``worst-case feasiblity'' with respect to functions observed up to $T$. 

    We remark that the parameter $\alpha$ measures the worst case feasibility of the problem by considering both budget and ROI constraints. In absence of the latter constraints, $\alpha$ would coincide with $\rho$. 
	We start by developing our analysis under the following standard assumption.

\begin{assumption}\label{assumption adv}
	In the adversarial (resp., stochastic) setting, $\gamma$ (resp., $\distr$) is such that $\alpha>0$.
\end{assumption}

This means that $\LP_\distr$ and $\LP_{\bar\gamma}$ satisfy (stochastic) Slater's condition. In particular, in the adversarial setting we are requiring the existence of a randomized ``safe'' policy that, in expectation, strictly satisfies the constraints for each $t$. This is a frequent assumption in works focusing on settings similar to ours (see, \eg \citep{chen2017online,neely2017online,yi2020distributed,castiglioni2022unifying}). In \Cref{sec:generalization} we show how this requirement can be relaxed.

When studying primal-dual algorithms, a key implication of Slater's condition is the existence and boundedness of Lagrange multipliers (see, \eg \citet{nedic2009approximate}). 
Therefore, when $\alpha>0$ is known, the set of dual multipliers can be safely bounded by requiring the $\ell_1$-norm of the multiplier to be less than or equal to $1/\alpha$ (see, \eg \citet{balseiro2022best}). 
This is the case, for example, for problems with only budget constraints, in which $\alpha=\rho>0$, which is achieved by bidding the void action $\nullx$ at each round. 
However, ROI constraints complicate the problem as the decision maker does \emph{not} know $\alpha$ beforehand. %

\section{Adaptivity in Primal-Dual Frameworks}\label{sec:primal dual}

In this section, we first provide a concise overview of a generic primal-dual template that adheres to the structure presented by \citet{immorlica2019adversarial,Castiglioni2022Online}.
Then, we provide a simple example demonstrating that a direct application of such framework would result in violations of the constraints which are linear in $T$. 
Finally, we describe the modifications needed to update the standard primal-dual template in order to achieve the desired behavior, and we show that online gradient descent already meets the new criteria for the dual regret minimizer.

\subsection{A Standard Primal-Dual Template}\label{sec:standard primal dual}

\cref{alg:meta alg} summarizes the structure of a standard primal-dual framework. It assumes access to two regret minimizers with the following characteristics. The first one is the \emph{primal regret minimizer} $\cRp$. It outputs policies in $\Pi$, and receives as feedback the loss $\lossp:\cX\to\R_{\ge0}$ such that, for any $x_t$ sampled according to policy $\vpi_t$, 
\begin{equation*}
    \lossp(x_t)= -f_t(x_t) + \lambda_t (g_t(x_t)+1)+ \mu_t (h_t(x_t)+1)+1.
\end{equation*}
The primal loss function is obtained from the Lagrangian relaxation of the problem at time $t$, plus the additive term $1+\lambda_t+\mu_t$ to ensure $\lossp(\cdot)\in\R_{\ge0}$. For ease of presentation, whenever we write $\lossp(\vpi)$ we mean $\lossp(\vpi)=\E_{x\sim\vpi}\lossp(x)$.
The second regret minimizer is the \emph{dual regret minimizer} $\cRd$. It outputs points in the space of dual variables $\cD_\alpha\defeq \{\vy\in\R^2_{\ge0}:\|\vy\|_1\le1/\alpha\}$, and observes the linear utility 
\begin{equation*}
   \lossd:\cD_{\alpha}\ni(\lambda,\mu) \mapsto \lambda g_t(x_t)+ \mu h_t(x_t)\in\R.
\end{equation*}
The dual regret minimizer has full feedback by construction. 

For each $t$, the algorithm first computes primal and dual actions. %
The action at time $t$ is $x_t\sim\vpi_t$ unless the available budget $B_t$ is less than 1, in which case we set $x_t$ equal to the void action $\nullx$. Then, $\lossp(x_t)$ and $\lossd$ are observed, and the budget consumption is updated according to the realized cost $c_t$. 
Finally, the internal state of the two regret minimizers is updated on the basis of the feedback specified by $\lossp,\lossd$. 
We will denote by $\tau\in[T]$ the time in which the budget is fully depleted and the decision maker starts playing the void action $\nullx$.

The traditional requirement on the primal and dual regret minimizers is that their cumulative regret should grow sublinearly in time (see, \eg \citet{Castiglioni2022Online}).\footnote{In general, the cumulative regret for losses $\ell_t$ is defined as $ \inf_{x \in\cX}\sum_{t=1}^T\mleft(\ell_{t}(x_t)-\ell_t(x)\mright)$ \citet{cesa2006prediction}.}
Although sufficient for handling simpler problems, the following example shows that this requirement is not enough to provide guarantees in our setting.

\begin{algorithm}[tb]
    \caption{Primal-dual framework.}
    \label{alg:meta alg}
 \begin{algorithmic}
    \STATE {\bfseries Input:} parameters $B,T,\delta$; regret minimizers $\cRp$, $\cRd$
    \STATE {\bfseries Initialization:} $ B_1\gets B$; initialize $\cRp,\cRd$
    \FOR{$t = 1, 2, \ldots , T$}
    \STATE {\bfseries Dual decision:} $(\lambda_t,\mu_t) \gets\textnormal{ output of }\cRd$
    \STATE{\bfseries Primal decision:} $\Pi\ni \pi_t\gets \textnormal{ output of }\cRp$

    \STATE {\bfseries Select action} as
    \vspace{-.3cm}
    \[
    x_t\gets \mleft\{\hspace{-1.25mm}\begin{array}{l}
            \displaystyle
            x_t \sim \vpi_t \hspace{.5cm}\text{\normalfont if } B_{t} \ge 1\\ [2mm]
            \nullx \hspace{.5cm}\text{otherwise}
        \end{array}\mright..
    \]
    \vspace{-.3cm}

    \STATE {\bfseries Observe:} observe $f_t(x_t)$ and $c_t(x_t)$, and update available resources: $B_{t+1}\gets B_{t} - c_t(x_t)$

    \STATE {\bfseries Primal update:} update $\cRp$ using $\lossp(x_t)$

    \STATE {\bfseries Dual update:} update $\cRd$ using $\lossd(\cdot)$
    \ENDFOR
 \end{algorithmic}
 \end{algorithm}

 \subsection{When Standard Primal-Dual Algorithms Fail}\label{sec:example}

We present a simple example in which the direct application of \Cref{alg:primal regret minimizer} with the standard requirements presented in \Cref{sec:standard primal dual} does not yield the desired behavior, even if the learner knows $\alpha$ a priori. We observe that this rule out the direct application of known primal-templates for adversarial inputs such as those by \citet{balseiro2022best,immorlica2019adversarial,Castiglioni2022Online}.

Suppose the learner is participating in a sequence of \emph{generalized first price auctions} (GFP) under an ROI constraint.
In this setting, there are multiple ad slots that have to be allocated at each $t$. The bidder with the $i$-th highest bid is allocated the $i$-th slot and, upon winning a slot, their payment is equal to their bid amount. 
\begin{table}[!t]
    \centering
    \begin{tabular}{c|c|c|c||cc}
      & Slot \circled{1}  & Slot \circled{2} & Slot \circled{3} & $b_t$ & $\mu_t$\\\toprule
     $v$ & 1 & $1/2$ & $1/16$ & &\\\cmidrule{1-4}
    $\cI_1$ & \cellcolor{gray!20} $1/2$  & \cellcolor{gray!20} $1/2$ & \cellcolor{gray!20} 0 & $1/2$ & 0\\ 
    $\cI_2$ & \cellcolor{gray!20} $1$  & \cellcolor{gray!20} $1/2$ & \cellcolor{gray!20} 0 & 0 & 0 \\
    $\cI_3$ & \cellcolor{gray!20} $1$  & \cellcolor{gray!20} $1/2$ & \cellcolor{gray!20} 0 & 1 & 16 \\
    \bottomrule
    \end{tabular}
    \caption{Setup of the repeated GFP auctions.}
    \label{tab:example}
    \vspace{-4mm}
\end{table}
\Cref{tab:example} provides a summary of the instance being considered. 
There are three ad slots \circled{1}, \circled{2}, and \circled{3} at each $t$. The valuation $v_i$ for each slot $i$ is fixed across the entire time horizon. Let $t_1,t_2\in [T]$, $t_1\le t_2$. We denote by $\cI\defeq[t_1,t_2]$ the set $\{t_1,t_1+1,\ldots,t_2\}$, and we call $\cI$ the \emph{time interval} starting from round $t_1$ to round $t_2$. We consider three time intervals, denoted by $\cI_1=[t_1]$, $\cI_2=(t_1,t_2]$, $\cI_3=(t_2,T]$ for some $t_1,t_2$. The cells highlighted in gray provide the highest competing bid for the different slots, in the three different time intervals. As an illustration, within interval $\cI_2$ the learner has to bid $1/2\le b_t <1$ to win \circled{2}. Within $\cI_1$, when the learner bids $b_t\ge 1/2$ they win \circled{1}. The learner has quasi-linear utilities: their utility for winning slot $i$ at time $t$ is $v_i-b_t$. The learner has a ROI constraint with target ROI 1.

The last two columns of \Cref{tab:example} describe a possible sequence of primal actions $b_t$ and dual actions $\mu_t$, which are constant within each interval. We observe that the dual variable is always at most $1/\alpha=16$. It is possible to show that, by setting the length of intervals so that $\li{1}=39\li{3}$ and $\li{2}\ge 12\li{3}$, the primal and dual regret over $T$ are $\le 0$, thereby matching the standard requirements on regret as per \Cref{sec:standard primal dual} (calculations are provided in \Cref{app:example}). However, the cumulative constraint violations is
\[
\sum_{t=1}^T\mleft(c_t(b_t)-f_t(b_t)\mright)=  \sum_{t=1}^T \mleft(2b_t-v_t\mright)=\Omega(T). 
\]
Therefore, the standard requirements for primal and dual regret minimizers are insufficient to ensure sublinear regret and constraint violations. The crucial problem here is that, since the primal player attains negative regret in $\cI_1$ and $\cI_2$, then it can afford to make decisions that significantly violate the ROI constraint in $\cI_3$. 
We observe that frameworks employing a recovery phase, such as the one by \citet{castiglioni2022unifying}, are not viable for our stated goals, since they rely on knowledge of $\alpha$ to switch between phases.

\subsection{New Requirements: Weak Adaptivity}

Unlike previous work, we require $\cRp$ and $\cRd$ to be \emph{weakly adaptive}, that is, they should guarantee sublinear \emph{adaptive} (\emph{a.k.a. interval}) regret (see, \eg \citep{hazan2007adaptive,luo2018efficient}). This notion of regret is stronger than ``standard'' external regret, and it will be essential in our analysis. 
The primal regret minimizer must be such that, for $\delta\in(0,1]$, with probability at least $1-\delta$ it holds that, for any $\vpi\in\Pi$ and for any interval $\cI$,
\begin{equation}\label{eq:interval regret} \sum_{t \in \cI}\mleft( \lossp(x_t)-\lossp(\vpi) \mright)\le M_{\cI}^2 \cump[T,\delta],
\end{equation}
where $M_{\cI}$ is the maximum absolute value of the losses $\lossp$ observed in interval $\cI$, and $\cump[T,\delta][\term{P}]$ is a term of order $\tilde O(\sqrt{T})$.
We require a similar property for the dual regret minimizer. However, since the dual regret minimizer works under full-information feedback by construction, we can use a regret minimizer with deterministic regret guarantees. In particular, $\cRd$ should guarantee that, for any time interval $\cI=[t_1,t_2]$, and for any pair of dual variables $(\lambda,\mu) \in \mathbb{R}_{\ge0}^2 $ it holds
\begin{equation*}\label{eq:interval dual}
\sum_{t \in \cI}\mleft(\lossd(\lambda,\mu) - \lossd(\lambda_t,\mu_t) \mright) \le
\nu(T)(\mu-\mu_{t_1})^2 + \cumdb+\cumdr,
\end{equation*}
where $\nu(T)\ge0$ is such that $\nu(T)=o(T)$, and $\cumdb$ (resp., $\cumdr$) is a term sublinear in $T$ related to the budget (resp., ROI) constraint.

The choice of the primal regret minimizer is primarily influenced by the specific problem being considered. On the other hand, the dual problem remains constant, and thus we proceed by presenting an appropriate dual regret minimizer. \Cref{sec:app} will provide some examples of primal regret minimizers satisfying \Cref{eq:interval regret} for relevant applications.

\subsection{A Weakly Adaptive Dual Regret Minimizer}\label{sec:dual rm}

As a dual regret minimizer we employ the standard online gradient descent algorithm (OGD) \citep{zinkevich2003online} on each of the two Lagrangian multipliers $\lambda$ and $\mu$. We initialize the algorithm by letting $\mu_1=\lambda_1=0$. We employ two separate learning rates $\etab$ and $\etar$, which will be specified in \cref{lm:dual rm}. At each round $t\in [T]$, the dual regret minimizer updates the Lagrangian multipliers as $\lambda_{t+1} \gets  P_{[0,1/\rho]}\mleft(\lambda_{t}+ \etab g_t(x_t)\mright)$, and $\mu_{t+1}\gets P_{\R_+}\mleft[\mu_{t}+  \etar h_t(x_t)\mright]$, where $P$ denotes the projection operator. The former update performs a gradient step and then projects the result on the interval $[0,1/\rho]$. This is possible because we know that playing the void action $\nullx$ would satisfy the budget constraints by a margin of at least $\rho$, and therefore we can safely consider as the set of $\lambda$ the interval $[0,1/\rho]$ \citep{Castiglioni2022Online}. On the other hand, in the update of $\mu$ we perform a grandient step and then ensure that the value is in $\R_+$. Since the decision maker does not know the feasibility parameter of ROI constraints, bounding $\mu$ becomes more complex, and we show how to approach this problem in \Cref{sec: bound multipliers}. 

Given a time interval $\cI=[t_1,t_2]$, and $\delta\in[0,1]$, we let
\begin{equation*}\label{eq:azuma term}
   \cume\defeq\mleft\{\begin{array}{lll}
       \displaystyle
       2\sqrt{(t_2-t_1)\log\mleft(2T/\delta\mright)} & \text{\normalfont if } \delta\in(0,1]  \\
       0 &  \text{\normalfont if } \delta=0\\
   \end{array}\mright.,
\end{equation*}
and,  when clear from context, we drop the dependency on $\cI$ to denote $\cume[T,\delta][[T]]$.
Let $\cumdb$ be a term of order $O(T^{1/2}/\rho)$, and $\cumdr$ be a term of order $O(T^{1/2})$.
The regret guarantees of the dual regret minimizer follow from standard results on OGD (see \citet[Chapter 10]{hazan2016introduction}).

\begin{lemma}\label{lm:dual rm}
    Let $\lambda_1=\mu_1=0$. 
    Then, OGD guarantees that, for any interval $\cI=[t_1,t_2]$, it holds 
    \begin{OneLiners}
    \item $\sum_{t \in \cI} \mu_t h_t(x_t)\le (\mu-\mu_{t_1})^2/\etar+\cumdr$ for  $\mu \in \R_+$,
    \item $\sum_{t \in \cI} \lambda_t g_t(x_t)\le \cumdb$ for all $\lambda \in R_+$,
    \end{OneLiners}
    where learning rates are set as follows: $\etab~\defeq~1/\rho T^{1/2}$, and $\etar~\defeq~1/\mleft(6+T^{1/2}+\cumdb+6\cume+16\cump[T,\delta]\mright)$.
\end{lemma}
The dependency on $\delta$ in the construction of $\etar$ is resolved by \cref{alg:meta alg} taking $\delta$ as an input parameter, and the final guarantees of the framework are parametrized on $\delta$.
Next, we prove the following simple result characterizing the growth of the $\mu$ variables, which will be useful in the remainder of the paper (omitted proofs can be found in the appendix).
\begin{restatable}{lemma}{growthMu} \label{lm:ogd}
    For all $t_1,t_2\in [T]$, it holds  \[\mu_{t_2}\ge \etar \sum_{t'\in[t_1,t_2-1]} h_{t'}(x_{t'})+\mu_{t_1}.\]
\end{restatable}

\section{Bounding the Lagrange Multipliers}\label{sec: bound multipliers}

Previous work usually assumes knowledge, either exactly or via some upper bound, of the Slater's parameter $\alpha$ \citep{balseiro2020dual,Castiglioni2022Online}. This information is then used to bound the magnitude of dual multipliers, which is fundamental in order to obtain meaningful primal regret upper bounds since the magnitude of $\lossp$ depends on dual multipliers.
In our setting, the decision maker has no knowledge of the gap that the strictly feasible solution guarantees for the ROI constraint, which renders the traditional approach to bound $\mu_t$ not viable. 
We show that, even without a priori information on $\alpha$, the primal-dual framework endowed with weakly adaptive regret minimizers guarantees that, with high probability, the Lagrange multiplier $\mu_t$ is bounded by $2/\alpha$ throughout the entire time horizon. 
We start by providing a general condition that we will prove to be satisfied both in the stochastic and adversarial setting.
\begin{definition}[$\delta$-safe policy]\label{def:safe policy}
Given $\delta\in (0,1]$, a policy $\safep$ is $\delta$-\emph{safe} if, for any interval $\cI\defeq [t_1,t_2]$, with $t_1,t_2\in [T]$, $t_1<t_2$, it holds
 \[
    \sum_{t \in \cI} \lambda_t g_t(\safep) + \mu_t h_t(\safep) \le \mleft(\mu_{\cI}+\nicefrac{1}{\alpha}\mright) \cume - \alpha \sum_{t \in \cI} \mu_t,
\]
where $\mu_\cI$ is the largest multiplier in the interval $\cI$. %
\end{definition}
A safe policy gives to the primal regret minimizer a way to limit the realized penalties imposed by the dual regret minimizer. In particular, we can show that if the dual regret minimizer increased the value of Lagrange multipliers $\mu_t$ too much, then the primal regret minimizer could ``fight back'' by playing the safe policy $\safep$, thereby preventing the dual player from being no-regret.
Indeed, the next result shows that whenever there exists a safe policy the Lagrange multipliers must be bounded. 

\begin{restatable}{theorem}{lemmaBoundLagrangian}\label{lm:boundLagrangian}
    If there exists a safe policy and the primal regret minimizer has regret at most $M_{\cI}^2 {\cump[T,\delta]}$ for any time interval $\cI$, then the Lagrange multipliers $\mu_t$ are such that $\mu_t \le 2/\alpha$ for each $t\in [\tau]$.
\end{restatable}

Then, we show that both in the stochastic and in the adversarial setting there exists a safe policy w.h.p., which implies that w.h.p. the Lagrange multipliers are bounded. 

\begin{restatable}{lemma}{LemmaSafeStoc}\label{lm:safeStoc}
If inputs $(f_t,c_t)$ are drawn i.i.d. from $\distr$, and there exists a policy $\vpi$ such that $\E_\distr g(\vpi)\le -\alpha$ and $\E_\distr h(\vpi) \le - \alpha$, then there exists a $\delta$-safe policy with probability at least $1-\delta$.
\end{restatable}

\begin{restatable}{lemma}{LemmaSafeAdv}\label{lm:safeAdv}
    If inputs $(f_t,g_t)$ are selected adversarially, and there exists a policy $\vpi$ such that $g_t(\vpi)\le -\alpha$ and $h_t(\vpi)\le -\alpha$ for each $t \in [T]$, then there exists a $\delta$-safe policy for any $\delta\in [0,1]$.
\end{restatable}

\section{Regret and Violations Guarantees}\label{sec:guarantees}

In this section, we describe guarantees for the stochastic and adversarial setting provided by \cref{alg:meta alg} equipped with a weakly adaptive primal and dual algorithm. Interestingly, we prove best-of-both-worlds guarantees through a unified argument which captures both cases. This is not the case in previous work, where the analysis of the stochastic case typically requires to study convergence to a Nash equilibrium of the \emph{expected Lagrangian game} \cite{immorlica2019adversarial}, which is not well defined in the adversarial setting.

We introduce the following event that holds w.h.p., most of our results will hold deterministically given this event.

\begin{definition}\label{def:event}
	We denote with $\CE$ the event in which \cref{alg:meta alg} satisfies the following conditions: i) the primal regret minimizer has regret upper bounded by $(3/\alpha+1) \cump[T,\delta]$ for all time intervals $\cI$, and ii) the dual multipliers for the ROI constraint are such that $\mu_t\le 2/\alpha$ for each $t \in [T]$.
\end{definition}
By applying a union bound to the events of \cref{eq:interval regret} and \cref{lm:safeStoc} or \cref{lm:safeAdv} (depending on the input model), we can use \cref{lm:boundLagrangian} to get the following result.
\begin{restatable}{lemma}{lmEvent}\label{lm:event high prob}
    Event $\CE$ holds with probability at least $1-2\delta$.
\end{restatable}

We start by observing that the cumulative violation of ROI constraints must be sublinear in $T$ with high probability, under both input models. This is a direct consequence of properties of the dual regret minimizer (see \cref{sec:dual rm}), and of the bound on dual multipliers implied by \cref{lm:event high prob}.

\begin{restatable}{lemma}{hbound}\label{lemma:small violations}
    If $\CE$ holds, then $\sum_{t \in [\tau]} h_t(\vpi_t)\le 1+2/(\etar\alpha)$.
\end{restatable}

Then, we define the following class of policies. 

\begin{definition}[$\mleft(\delta,q,\OPT\mright)$-optimal policy]\label{def: good policy}
    Given $\delta\in(0,1]$, $q\in(0,1]$, a sequence of $T$ inputs $\{(f_t,c_t)\}_{t=1}^T$, and the value of a baseline $\OPT$, we say that a policy $\vpi$ is \emph{$\mleft(\delta,q,\OPT\mright)$-optimal}, if i) $\sum_{t\in [T]} f_t(\vpi)\ge q\cdot T\cdot \OPT - \cumg$, and ii) $\sum_{t \in [t']} \lambda_t g_t(\vpi)  +\mu_t h_t(\vpi)\le \mleft(\mu_{[t']}+\frac{1}{\alpha}\mright) \cumg $ for each $t' \in [T]$, where $\mu_{[t']}$ is the largest multiplier $\mu_t$ observed up to $t'$. 
\end{definition}

A $\mleft(\delta,q,\OPT \mright)$-optimal policy guarantees a reward which is a fraction $q$ of the reward of the baseline up to a sublinear term, and guarantees that the cumulative value of the penalty due to the Lagrangian relaxation is vanishing in time.

First, we need the following result that holds both in the stochastic and adversarial setting.

\begin{restatable}{lemma}{lemmaBeforeEnd} \label{lm:highG}
    \Cref{alg:meta alg} guarantees that 
    $\sum_{t \in [\tau]} \lambda_t g(x_t)\ge T-\tau-1/\rho-\cumdb.$
\end{restatable}
Given primal (resp., dual) regret minimizer with guarantees $\cump[T,\delta]$ (resp., $\cumdr$ and $\cumdb$), and $\delta\in(0,1]$, we let 
\[
\reg(T,\alpha,\delta)\defeq \nicefrac{1}{\alpha}  +\mleft(\nicefrac{3}{\alpha}+1\mright) \mleft(\cump[T,\delta]+\cumg\mright)+\cumdr+\cumdb.
\]
The existence of a $(\delta,q,\OPT)$-optimal policy implies the following bound with respect to the generic baseline $\OPT$.

\begin{restatable}{lemma}{lemmaRegretUnif}\label{lm:regret unified}
    Suppose event $\CE$ holds and that there exists a $\mleft(\delta,q,\OPT\mright)$-optimal policy. Then,
    $
        \textstyle\sum_{t \in [\tau]} f_t(x_t)\ge  q T\OPT - \reg(T,\alpha,\delta).
    $ 
\end{restatable}

\vspace{-.1cm}
Next, we show that a suitable $\mleft(\delta,q,\OPT\mright)$-optimal policy exists w.h.p. both in the stochastic and adversarial setting. 

\begin{restatable}{lemma}{LemmaOptStoc}\label{lm:optimal policy stoc}
    In the stochastic setting, with probability at least $1-2\delta$ there exists a $(\delta,1,\OPT_\distr)$-optimal policy (where $\OPT_\distr$ is the optimal value of \ref{eq:opt lp gen}).
\end{restatable}
This is saying that, given a distribution $\distr$, there exists with high probability a policy satisfying \Cref{def: good policy}. 
Stochasticity of the environment 
is used to prove that the solution to \ref{eq:opt lp gen} satisfies the second condition of \Cref{def: good policy} for \emph{all} $t\in[T]$. 
If we tried a similar approach in the adversarial setting, the solution to $\term{LP}_{\bar\gamma}$ would guarantee that the second condition is satisfied over the whole time horizon, but \emph{not} necessarily at earlier time steps $t<T$. Moreover, feasibility in expectation has no implications on feasibility of a policy under the adversarial sequence $(\lambda_t,f_t,\mu_t,c_t)$, in which dual variables are optimized to ``punish violations''.
However, it is possible to show the existence of a policy satisfying \Cref{def: good policy} even in the adversarial setting via a different approach. 
We build a suitable convex combination between a strictly feasible policy $\safep$ guaranteeing that constraints are satisfied by at least $\alpha>0$ for each $t\in [T]$, and the optimal unconstrained policy $\vpi^\ast$ maximizing $\sum_{t\in [T]}f_t(\vpi)$. The following lemma employs a policy $\hat\vpi$ such that, for all $x\in\cX$, $\hat\vpi_x\defeq \safep_x/(1+\alpha) + \alpha\vpi^\ast_x/(1+\alpha)$.

\vspace{.1cm}
\begin{restatable}{lemma}{LemmaOptAdv}\label{lm:optimal policy adv}
    In the adversarial setting, there always exists a $(0,\alpha/(1+\alpha),\OPT_{\bar \gamma})$-optimal policy.
\end{restatable}

Now, we provide the overall guarantees of the dual-balancing primal-dual algorithm (\cref{alg:meta alg}). 

\begin{theorem}[Stochastic setting]\label{thm:main stoc}
    In the stochastic setting, for $\delta>0$, with probability at least $1-4\delta$ \Cref{alg:meta alg} guarantees 
    $
       \textstyle T\OPT_\distr - \sum_{t\in [T]}f_t(x_t)\le \reg(T,\alpha,\delta).
    $
    Moreover, we have the following guarantees on constraint violations: $\sum_{t \in [T]} h_t(x_t)\le 1+2/(\etar\alpha)$ and $\sum_{t\in[T]}c_t(x_t)\le B$.
\end{theorem}
\begin{proof} The regret upper bound holds since event $\CE$ holds with probability at least $1-2\delta$ (\cref{lm:event high prob}), and by combining \cref{lm:optimal policy stoc} and \cref{lm:regret unified}. The ROI constraint is upper bounded by \cref{lemma:small violations}, and the budget constraint is strictly satisfied by construction of \cref{alg:meta alg}.
\end{proof}
Since the primal regret minimizer guarantees a high-probability primal regret upper bound of order $T^{1/2}$ (see \Cref{eq:interval regret}), then the cumulative regret and the cumulative ROI constraint violation of \Cref{thm:main stoc} are of order $\tilde O(\sqrt{T})$, while the budget constraint is strictly satisfied.

Analogously, by exploiting \cref{lm:optimal policy adv}, we have the following guarantees for the adversarial setting.
\begin{theorem}[Adversarial setting]\label{thm:main adv}
    Suppose the sequence of inputs $\gamma=(f_t,c_t)_{t=1}^T$ is selected by an oblivious adversary. Then, for $\delta> 0$, with probability at least $1-2\delta$, \Cref{alg:meta alg} guarantees
   $
    \textstyle \frac{\alpha}{1+\alpha}\OPT_{\bar \gamma} - \sum_{t \in [T]} f_t(x_t)\le \reg(T,\alpha,\delta).~\footnote{The same guarantees would hold with respect to the optimal unconstrained policy maximizing $\sum f_t(\vpi)$.}
    $
    Moreover, it holds that $\sum_{t \in [T]} h_t(x_t)\le 1+2/(\etar\alpha)$ and $\sum_{t\in[T]}c_t(x_t)\le B$.
\end{theorem}
Our competitive ratio matches that of \citet{castiglioni2022unifying}, and, in the case in which we only have budget constraints, it yields the state-of-the-art $1/\alpha$ competitive ratio of \citet{Castiglioni2022Online}.

\section{Relaxing the Safe-Policy Assumption}\label{sec:generalization}

In the adversarial setting, the usual assumption for recovering Slater's condition is that there exists a policy guaranteeing that constraints are satisfied by at least $\alpha>0$ for each $t$ \citep{chen2017online,yi2020distributed,castiglioni2022unifying}. Our analysis, up to this point, made the same assumption (\cref{assumption adv}), except that, unlike those past works, we do not need to know the value of $\alpha$. Now, we show that our analysis carries over with the following looser requirement.
\begin{assumption}\label{assumption new}
    There exists a policy $\safep\in\Pi$ such that, for each interval $\cI=[t_1,t_2]$ with $t_2-t_1= k$, we have $\sum_{t \in \cI} g_t(\safep)\le -\alpha k$ and $\sum_{t \in \cI} h_t(\safep)\le -\alpha k$.
\end{assumption}
The traditional assumption of requiring a safe policy for \emph{each} $t$ would require the decision maker to have an action yielding expected ROI strictly above their target for each round $t$. 
This may not hold in practice. For example, in the case of repeated ad auctions, if we assume one ad placement is being allocated at each $t$, then the agent would be priced out by other bidders for at least some time steps.
Next, we show that if the size of the intervals $k$ is not too big (\ie if there exists a ``safe'' policy frequently enough), there exist the following policies.
\begin{restatable}{lemma}{newAssumption}\label{lm:new assumption}
    Suppose \cref{assumption new} holds with $k<\cumg/(2T\etab)$. Then, for $\delta>0$, there exists a $\delta$-safe and a $(\delta,\alpha/(1+\alpha),\OPT_{\bar\gamma})$-optimal policy.
\end{restatable}
This allows us to balance the tightness of the required assumption with the final regret guarantees, by suitably choosing the learning rates rates $\etab$ and $\etar$. 
When \cref{assumption new} holds for $k=\log T$, we recover exactly the bounds of \Cref{thm:main adv}. As a further example, if $k=T^{1/4}$, then we can obtain regret guarantees of order $\tilde O(T^{3/4})$ by setting $\etab=O(T^{-3/4})$ and 
by suitably updating the definition of $\cumg$.
In the context of ad auctions, this allows us to make the milder assumption that the bidder sees an auction with ROI $>0$ at least every $k$ steps, instead of at every step.

\begin{algorithm}[tb]
    \caption{Primal regret minimizer.}
    \label{alg:primal regret minimizer}
 \begin{algorithmic}
    \STATE {\bfseries Input:} parameters $\eta>0,\xi>0,\sigma>0$
    \STATE {\bfseries Initialization:} $[0,1]^{\nval\times\nbid} \ni \vw_{1}\gets \vone$ 
    \FOR{$t = 1, 2, \ldots , T$}
    \STATE
    \begin{OneLiners}
     \item {\bfseries Observe} valuation $v_t\in\cV$
     \item {\bfseries Set} $\vpi(v_t)_x \gets w_{t,v_t,x} / \sum_{x'\in\cX} w_{t,v_t,x'},\,\,\forall x\in \cX$
     \item {\bfseries Bid} $x_t \sim \vpi(v_t)$
     \item {\bfseries Observe} loss $\lossp(x_t)$
     \item $\tilde \lossp(x)\gets \lossp(x)\indicator{x=x_t}/(\vpi(v_t)_x+\xi)$ $\forall x\in\cX$
     \item For each $x\in\cX$, set $w_{t+1,v_t,x} \gets (1-\sigma) w_{t,v_t,x}\cdot e^{-\eta\tilde \lossp(x)} 
     + \frac{\sigma}{\nbids} \sum_{x'\in\cX} w_{t,v_t,x'} \cdot e^{-\eta\tilde \lossp(x')}$
 \end{OneLiners}
    \ENDFOR
 \end{algorithmic}
 \end{algorithm}

\section{Bidding in Repeated Non-Truthful Auctions}\label{sec:app}

In automatic bidding systems advertisers usually have to specify some parameters like their overall budget and their targeting criteria. Then a \emph{proxy bidder} operated by the platform places bids on their behalf. 
A popular autobidding strategy is \emph{value maximization} subject to budget and ROI constrains \citep{auerbach2008empirical,golrezaei2021auction,deng2023multi}.
Recently, many advertising platforms have been transitioning from the second-price auction format toward a first-price format (see, \eg \citep{admanager,wong21firstprice}), which is \emph{not} truthful. %
In this context, existing results for truthful auctions cannot be applied~\citep{balseiro2019learning,feng2022online}.

We show that our framework can be used to manage bidding in repeated non-truthful auctions under budget and ROI constraints. We will focus on the case of repeated first-price auctions, and we will make the simplifying assumption of having a finite set of possible valuations and bids. 
In \Cref{app:app}, we provide further details on this application and also apply the framework to second-price auctions with continuous bids and valuations, extending the stochastic analysis by \citet{feng2022online} to adversarial settings.

\xhdr{Set-up.} 
At each round $t\in [T]$, the bidder observes their valuation $v_t$ extracted from a finite set $\cV\subset [0,1]$ of $n$ possible valuations. The set $\cX\subset [0,1]$ is a set of $\nbid$ possible bids. Let $\beta_t$ be the highest-competing bid at time $t$. In the value-maximizing utility model for each $t$ we have $f_t(x_t)\defeq v_t \indicator{x_t\ge \beta_t}$ \citep{babaioff2021non,balseiro2021landscape},  and the cost function is $c_t(x_t)\defeq x_t\indicator{x_t\ge \beta_t}$, where the indicator function specifies whether the bidder won the auction at time $t$.
We extend the definition of policies from \Cref{sec:baselines} to model \emph{randomized bidding policies}. Each  $\vpi\in\Pi$ is now a mapping $\vpi:\cV\rightarrow \Delta_\cX$. We denote by $\vpi(v)_x$ the probability of selecting $x$ under valuation $v$.

\xhdr{Primal regret minimizer.} %
Our primal regret minimizer is based on the  EXP3-SIX algorithm by~\citet{neu2015explore} and it is described in \Cref{alg:primal regret minimizer}.
At each round $t$, the algorithm maintains a set of weights $\vw_t\in[0,1]^{n\times m}$. The probability of playing $x$ under valuation $v_t$ is proportional to the weight $w_{t,v_t,x}$. 
After drawing $x_t$, the algorithm observes $\lossp(x_t)$ and builds the estimated loss $\tilde\lossp$, where $\xi>0$ is the implicit exploration term. Then, the update of weights $\vw$ is inspired by the Fixed Share algorithm by \citet{herbster1998tracking}.
We start by showing \cref{eq:interval regret} holds in the single-valuation setting.

\begin{restatable}{theorem}{primalreg}\label{thm:adaptive regret primal short}
    Let $\nval=1$, $\eta\defeq 1/\sqrt{\nbid T}$, $\xi\defeq 1/(2\sqrt{\nbids T})$, $\sigma\defeq 1/T$. For any $\delta>0$, \emph{EXP3-SIX} guarantees that, w.p. at least $1-\delta$, for any interval $\cI$, and for any $x\in\cX$, 
    \[
        \textstyle \sum_{t\in \cI} \mleft( \lossp(x_t) - \lossp(x)\mright)\le O\mleft(\sqrt{mT}\log\mleft(\frac{\nbids T}{\delta}\mright)\mright).
    \]
\end{restatable}
Then, if we instantiate one independent instance of EXP3-SIX for each valuation in $\cV$ with the choice of parameters of \cref{thm:adaptive regret primal short}, we have that for any time interval $\cI$ the regret accumulated by \Cref{alg:primal regret minimizer} over $\cI$ is upper bounded by $  M^2_\cI\,\sqrt{\nval}\,\cump[T,\delta]$ with probability at least $1-\nval\delta$ (see \citet[Chapter 18.4]{lattimore2020bandit}). It follows that \Cref{eq:interval regret} is satisfied and the guarantees of \Cref{thm:main stoc,thm:main adv} readily apply to the problem of bidding in repeated first-price auctions under budget and ROI constraints.

\bibliography{refs}
\bibliographystyle{icml2024}

\newpage
\appendix
\onecolumn

\section{Further Related Works}\label{sec:apprelated}

We survey the most relevant works with respect to ours. For further background on online learning the reader can refer to the monograph by \citet{cesa2006prediction}.

\xhdr{1) Bandits with Knapsacks.} The stochastic \emph{Bandits with Knapsacks} (BwK) framework was introduced and first solved by~\citet{Badanidiyuru2018jacm}.
Other regret-optimal algorithms for stochastic BwK have been proposed by \citet{agrawal2019bandits}, and by \citet{immorlica2019adversarial}.
The BwK framework has been subsequently extended to numerous settings such as, for example, more general notions of resources and constraints~\citep{agrawal2014bandits,agrawal2019bandits}, contextual bandits~\citep{dudik2011efficient,badanidiyuru2014resourceful,agarwal2014taming,agrawal2016efficient}, and combinatorial semi-bandits \citep{sankararaman2018combinatorial}.
Moreover, the BwK framework has been employed to model various applications with budget/supply constraints such as, for example, dynamic pricing \citep{besbes2009dynamic,besbes2012blind,wang2014close}, dynamic procurement \citep{badanidiyuru2012learning}, and dynamic ad allocation \citep{combes2015bandits, balseiro2019learning}.
The \emph{Adversarial Bandits with Knapsacks} (ABwK) setting was first studied by \citet{immorlica2019adversarial}, who proved a $O(m\log T)$ competitive ratio for the case in which the sequence of rewards and costs is chosen by an oblivious adversary.
\citet{immorlica2019adversarial} also show that no algorithm can achieve a competitive ratio better than $O(\log T)$ on all problem instances, even in instances with only two arms and a single resource. 
Recently, \citet{kesselheim2020online} refined the analysis for the general ABwK setting to obtain an $O(\log m\,\log T)$ competitive ratio. They also prove that such competitive ratio is optimal up to constant factors.
Moreover, \citet{Castiglioni2022Online} proved a constant-factor competitive ratio in the regime $B=\Omega(T)$.
We mention that further results have been obtained in the simplified setting with one constrained resource \cite{rangi2018unifying,tran2010epsilon,tran2012knapsack}.
All the works mentioned in this paragraph can only handle packing constraints (\eg budget constraints). They cannot handle ROI constraints, and they need perfect knowledge of the feasibility parameter $\alpha$.

\xhdr{2) Online packing problems.} Various well-known online packing problems can be seen as special cases of ABwK, with a more permissive feedback model which allows the decision maker to observe the full feedback before choosing an action (see, e.g., \citet{Buchbinder2009design,devanur2011near}). 
In online packing settings, since the decision maker is endowed with more information at the time of taking decisions, it is possible to derive $O(\log T)$ competitive ratio guarantees against the optimal dynamic policy. 
In the context of online allocation problems with fixed per-iteration budget,~\citet{balseiro2020dual,balseiro2022best} propose a class of algorithms which attain asymptotically optimal performance in the stochastic case, and they attain an asymptotically optimal (parametric) constant-factor competitive ratio when the inputs are adversarial. In their setting, as we already mentioned, at each round the input $(f_t,c_t)$ is observed by the decision maker \emph{before} they make a decision. This makes the problem essentially different from ours.
Even in this case, these works cannot handle problems with ROI-constrained decision makers, since they can handle only packing constraints, and require knowledge of the feasibility parameter.

\xhdr{3) Online convex optimization with time-varying constraints.} Another line of related work concerns online convex optimization with time-varying constraints (see, \emph{e.g.}, \cite{mahdavi2012trading,mahdavi2013stochastic,jenatton2016adaptive,neely2017online,chen2018bandit,castiglioni2022unifying}), where it is usually assumed that the action set is a convex subset of $\R^m$, in each round rewards (resp., costs) are concave (resp., convex), and most importantly, resource constraints only apply at the last round. In contrast, in our setting, budget constraints apply in all rounds. Moreover, guarantees are usually provided either for stochastic constraints \cite{yu2017online,wei2020online}, or for adversarial constraints \cite{mannor2009online,sun2017safety,liakopoulos2019cautious}, typically by employing looser notions of regret. In contrast, our framework will provide best-of-both-worlds guarantees. Moreover, these frameworks typically require perfect knowledge of $\alpha$ ad \Cref{assumption adv} to hold, while our framework relaxes both assumptions.

\xhdr{4) Bidding in repeated auctions.} The problem of online bidding in repeated auctions has been extensively studied using online learning approaches (see, \eg \citet{borgs2007dynamics,weed2016online,nedelec2022learning}). 
In particular, online bidding under budget constraints has been studied in various settings. \citet{balseiro2019learning} and \citet{ai2022no} focus on utility-maximizing agents with one resource-consumption constraint. In the context of online allocation problems with an arbitrary number of constraints, \citet{balseiro2020dual,balseiro2022best} propose a class of primal-dual algorithms attaining asymptotically optimal performance in the stochastic and adversarial case. In their setting, at each round, the input $(f_t,c_t)$ is observed by the decision maker \emph{before} they make a decision. This makes the problem substantially different from ours. In particular, their framework cannot handle non-truthful repeated auctions. 
Recent works have also examined settings similar to ours, involving bidders with constraints on their budget and ROI. The framework by \citet{feng2022online} can handle ROI and ``hard'' budget constraints, but crucially relies on truthfulness of second-price auctions, and on the stochasticity of the environment.
The framework by \citet{castiglioni2022unifying} allows for general ``soft'' constraints under both stochastic and adversarial inputs. Their framework cannot be applied in our setting for three reasons: i) we have hard budget constraints, ii) we don't make the stringent assumption of knowing the parameter $\alpha$ beforehand, and iii) we relax the assumption of having one strictly feasible solution for each round in the adversarial setting. 
Finally, \citet{golrezaei2021bidding} studies the dynamic pricing problem faced by a seller who repeatedly sells items to a single budget and ROI constrained buyer.

\section{Further Details on the Example of \Cref{sec:example}}\label{app:example}

We have to set the length of $\cI_1$, $\cI_2$, and $\cI_3$ so that the primal and dual regret are $\le 0$, while the constraint violations are $\Omega(T)$.
We observe that, by construction, the two candidate actions for being the best in hindsight in the primal problem are bidding either $0$ or $1/2$. Therefore, we start by computing the primal regret with respect those two actions. For simplicity, we drop the rescaling factor from the definition of $\lossp$ since that is only needed for technical reasons in the new construction, and write: 
\begin{equation*}
    \lossp(b_t)= -f_t(b_t) + \mu_t h_t(b_t)= (c_t(b_t)-v_t) + \mu_t\,(2c_t(b_t)-v_t).
\end{equation*}
Then, the regret with respect to bid $1/2$ is
\[
\cump[T](1/2)=\sum_{t=1}^T\mleft(\lossp(b_t)-\lossp(1/2)\mright)=-\frac{\li{1}}{2} - \frac{\li{2}}{16} + 16\li{3} + \frac{\li{1}}{2} - \frac{\li{2}}{4}  - \frac{49\li{3}}{4}=-\frac{5\li{2}}{16}+\frac{15\li{3}}{4},
\]
and the regret with respect to bid $0$ is
\[
\cump[T](0)=\sum_{t=1}^T\mleft(\lossp(b_t)-\lossp(0)\mright)=-\frac{\li{1}}{2} - \frac{\li{2}}{16} + 16\li{3}+\frac{\li{1}}{16}+\frac{\li{2}}{16}+\frac{17\li{3}}{16}=-\frac{7\li{1}}{16}+\frac{273\li{3}}{16}.
\]
We observe that $\cump[T](0)\le 0$ for $\li{1}$ big enough. Interval $\cI_1$ can arbitrarily long since during the first phase the ROI violation is 0, so it does not impact on the constraints. In particular, we can set $\li{1}=39\li{3}$. Moreover, we observe that $\cump[T](1/2)\le 0$ whenever $\li{2}\ge 12\li{3}$. By setting $\li{2}= 12\li{3}$, we obtain $\li{3}=T/52$. 

The sequence of dual multipliers $\mu_t$ depicted in \Cref{tab:example} guarantees no-regret on the dual problem since the dual player is exactly best responding to the primal actions: $\mu_t= 0$ when violations are $\le 0$, and $\mu_t=1/\alpha=16$ when violations are strictly positive. We observe that in this example $\alpha=\frac{1}{16}$ since bidding $0$ yields a strict feasibility gap of at least $1/16$. Therefore, as expected, the two best-response actions of the dual player are the two extreme points of the interval $[0,1/\alpha]$ (see \Cref{sec:standard primal dual}).  

Now, we can write the cumulative violation as a function of $T$. We have
\[
\sum_{t=1}^T h_t(b_t)=\sum_{t=1}^T(2c_t(b_t)-v_t)=-\frac{1}{16}\li{2}+\li{3}=\frac{1}{208}T=\Omega(T).    
\]
Therefore, ROI constraint violations grow linearly in $T$, thereby violating one of our desiderata.

\section{Proofs for \cref{sec:primal dual}}

\growthMu*
\begin{proof}
    We prove the result by induction. 
    Fix a starting point $t_1\in [T]$.
    First, it's easy to see that the result holds for $t_2=t_1$.
    Then, suppose that the statement holds for round $t_2=t$. Then,
    \begin{align*}
        \mu_{t+1}&= \mleft[ \mu_t + \etar\, h_t(x_t)\mright]^{+}\\
        &\ge \mu_t + \etar\, h_t(x_t)\\
        &\geq \etar \sum_{t' \in [t_1,t-1]} \, h_{t'}(x_{t'}) + \mu_{t_1}+ \etar\, h_{t}(x_{t})\\
        &= \etar \sum_{t' \in [t_1,t]}  h_{t'}(x_{t'}) + \mu_{t_1},
    \end{align*}
    where $[x]^+\defeq\max\{x,0\}$.
    This implies that the statement holds for $t_2=t+1$.
    This concludes the proof.
\end{proof}

\section{Proofs for \cref{sec: bound multipliers}}

\lemmaBoundLagrangian*

\begin{proof}
    We consider two cases. 

    \xhdr{Case 1:} $\alpha\le 10/\sqrt{T}$. By construction of the dual regret minimizer, and by the choice of $\etar$, the dual variable $\mu_t$ can reach at most value $\etar T\le \sqrt{T}/16$. Therefore, we have $\mu_t\le\sqrt{T}/16<2/\alpha$.

    \xhdr{Case 2:} $\alpha> 10/\sqrt{T}$. Let $\cI=[t_1,t_2]$, with $t_1,t_2\in [T]$, $t_1\le t_2$. Moreover, assume that there exists a safe policy $\safep$. We show that, if the Lagrangian multiplier $\mu_t$ is greater than $2/\alpha$, we reach a contradiction.

    Suppose, by contradiction, that there exists a round $t_2$ such that $\mu_{t_2}\ge 2/\alpha$.
    Let $t_1$ be the the first round such that $\mu_t\ge 1/\alpha$ for any $t\in[t_1,t_2]$. 
    Notice that the structure of the dual regret minimizer (see \cref{sec:dual rm}) implies that
    \begin{equation}\label{eq:mu t1 mu t2}
        \mu_{t_1}\le 1/\alpha+\etar \quad\text{\normalfont and }\quad \mu_{t_2}\le 2/\alpha+\etar,
    \end{equation} 
    since the dual losses are in $[-1,1]$.
    Therefore, we can upperbound the primal loss function as $M_{[T]}\le (1+4/\alpha)$ (\ie we use $\mu_t\le 3/\alpha$).
    This implies that, for $m\ge 2$, 
    \[\eta M_{[T]}=\frac{1}{\sqrt{mT}}\mleft(1+\frac{4}{\alpha}\mright)<\frac{1}{\sqrt{mT}}\mleft(1+\frac{2\sqrt{T}}{5}\mright)\le1.
    \]
    Therefore, the primal regret minimizer satisfies the bound on the adaptive regret of \cref{eq:interval regret}.
Then, by the no-regret property of the primal we get:
    \begin{align*}
        &\sum_{t\in\cI} \mleft(f_t(x_t) - \lambda_t g_t(x_t)- \mu_t h_t(x_t)\mright)\\ &\hspace{2cm}\ge \sum_{t\in\cI} (f_t(\safep)-\lambda_t g_t(\safep)- \mu_t h_t(\safep)) - M^2_{\cI} {\cump[T,\delta]} \\
        &\hspace{2cm}\ge \alpha  \sum_{t\in\cI}  \mu_t - \mleft(\mu_{\cI}+\frac{1}{\alpha}\mright) \cume -  M^2_{\cI} {\cump[T,\delta]} & \text{\normalfont (by \cref{def:safe policy})} \\
        &\hspace{2cm} \ge (t_2-t_1) - \mleft(\mu_{[t_1,t_2-1]}+\frac{3}{\alpha}+\etar\mright) \cume  -  M^2_{\cI} {\cump[T,\delta]} & \text{\normalfont (by Def. of $t_1$ and \cref{eq:mu t1 mu t2})}\\
        &\hspace{2cm} \ge(t_2-t_1) -  \mleft(\frac{5}{\alpha}+\etar\mright) \cume -  \mleft(1+\frac{2}{\alpha} + \frac{1}{\alpha}\mright)^2 {\cump[T,\delta]} & \text{\normalfont (by Def. of $M_\cI$ and $\lossp$)}\\
        &\hspace{2cm} \ge (t_2-t_1) - \mleft(\frac{5}{\alpha}+\etar\mright) \cume -  \frac{16}{\alpha^2} {\cump[T,\delta]}.\numberthis\label{eq:intermediate1}
    \end{align*}

Since the Lagrangian multipliers $\mu_t$ is always at least $1/\alpha$ for $t\in [t_1,t_2]$, the dual regret minimizer never has to project over $\R_{\ge 0}$ during interval $[t_1,t_2]$. In particular, projecting the dual multiplier at $t$ back onto $\R_{\ge 0}$ would yield $\mu_t=0$. This cannot happen for $t\in [t_1,t_2]$, since $\mu_t\ge 1/\alpha$. Then, since the dual regret minimizer does not perform any projection operation during $[t_1,t_2]$, we have that the statement of \cref{lm:ogd} holds with equality:
\[
    \mu_{t_2}= \etar \sum_{t'\in[t_1,t_2-1]} h_{t'}(x_{t'})+\mu_{t_1}.
\]
Hence, by definition of $t_2$ and \cref{eq:mu t1 mu t2},
\[
    \sum_{t'\in[t_1,t_2-1]} h_{t'}(x_{t'})= \frac{\mu_{t_2}-\mu_{t_1}}{\etar}\ge  \frac{1}{\alpha \etar} -1.
\]

Then, by the regret bound of the dual with respect to $\mu=\mu_t$ and $\lambda=0$, we get
\begin{align*} 
\sum_{t \in [t_1,t_2-1]}\mleft(f_t(x_t)-\lambda_t g_t(x_t)- \mu_t h_t(x_t)\mright) &\le \sum_{t \in [t_1,t_2-1]}\mleft(f_t(x_t)- \mu_{t} h_t(x_t)\mright)+  \cumdr + \cumdb\\
&\le (t_2-t_1) - \frac{1}{\alpha}  \sum_{t\in[t_1,t_2-1]} h_{t}(x_t) +  \cumdr + \cumdb\\
&\le (t_2-t_1) -  \frac{1}{\alpha^2 \eta}+ \frac{1}{\alpha} + \cumdr + \cumdb.\numberthis\label{eq:intermediate2}
\end{align*}

By putting \cref{eq:intermediate1} and \cref{eq:intermediate2} together we have that
\[
    (t_2-t_1) -  \frac{1}{\alpha^2 \etar}+ \frac{3}{\alpha} +  2 + \cumdr + \cumdb \ge  (t_2-t_1) - \mleft(\frac{5}{\alpha}+\etar\mright) \cume -  \frac{16}{\alpha^2} {\cump[T,\delta]}. 
\]
We observe that in \cref{lm:dual rm} we set $\etar\defeq \mleft(6+\cumdr+\cumdb+6\cume+16\cump\mright)^{-1}$. Then, from the inequality above we have
\[
    \frac{1}{\alpha^2 \etar} \le  \frac{3}{\alpha} + 2 + \cumdr + \cumdb  + \mleft( \frac{5}{\alpha}+1\mright)\cume + \frac{16}{\alpha^2} \cump[T,\delta].
\]
However, we reach a contradiction since 
\begin{align*}
    \frac{1}{\alpha^2\etar}&\ge \frac{4}{\alpha} + 2 +\cumdr+\cumdb+\mleft(\frac{5}{\alpha}+1\mright)\cume +\frac{16}{\alpha^2}\cump \\&>\frac{3}{\alpha} + 2 + \cumdr + \cumdb  + \mleft( \frac{5}{\alpha}+1\mright)\cume + \frac{16}{\alpha^2} \cump,
\end{align*}
where we used the fact that $\alpha\in (0,1]$ by assumption ($\alpha>0$), and by boundedness of $g_t$ and $h_t$ for all $t\in [T]$.
This concludes the proof.
\end{proof}

\LemmaSafeStoc*

\begin{proof}
    By the definition of $\alpha$, there exists a policy $\vpi$ such that $\E_\distr g(\vpi)\le -\alpha$ and $\E_\distr h(\vpi) \le - \alpha$. 
    Then, given a time interval $\cI=[t_1,t_2]$, $t_1,t_2 \in [T]$, by appling the Azuma–Hoeffding inequality to the martingale difference sequence $W_1,\ldots,W_T$ with  
    \[
       W_t\defeq\lambda_t g_t(\vpi) + \mu_t h_t(\vpi) - \lambda_t \E_\distr g(\vpi)-\mu_t \E_\distr h(\vpi),
    \]    
    we obtain that
    \[
        \mleft|\sum_{t\in\cI} W_t\mright| \le \mleft(\mu_{\cI}+ \frac{1}{\alpha}\mright) \sqrt{2(t_2-t_1) \log\mleft(\frac{2}{\delta}\mright)}
    \]
    holds with probability at least $1-\delta$.
    By applying a union bound we get that the inequalities for each time interval $\cI$ hold simultaneously with probability at least $1-T^2\delta$. Let $\cume\defeq 2\sqrt{(t_2-t_1)\log\mleft(\frac{2T}{\delta}\mright)}$ as per \cref{def:safe policy}.
    Then, with probability at least $1-\delta$ it holds 
    \begin{align*}
    \sum_{t \in \cI} \mleft(\lambda_t g_t(\vpi) + \mu_t h_t(\vpi) \mright) & \le \mleft(\mu_{\cI}+ \frac{1}{\alpha}\mright) \cume + \sum_{t \in \cI} \mleft( \lambda_t \E_\distr g(\vpi) + \mu_t \E_\distr h(\vpi)\mright) \\
    &\le \mleft(\mu_{\cI}+ \frac{1}{\alpha}\mright) \cume - \alpha \sum_{t \in \cI} \mleft( \lambda_t + \mu_t\mright) \\
    &\le \mleft(\mu_{\cI}+ \frac{1}{\alpha}\mright) \cume - \alpha \sum_{t \in \cI} \mu_t.
    \end{align*}
    This concludes the proof.
\end{proof}

\LemmaSafeAdv*

\begin{proof}
    By assumption there exists a policy $\vpi$ such that $g_t(\vpi)\le -\alpha$ and $h_t(\vpi)\le -\alpha$ for each $t \in [T]$. Then, for each $t_1,t_2 \in [T]$, with $t_1<t_2$, it holds
    $\sum_{t \in [t_1,t_2]} \mleft( \lambda_t g_t(\vpi) + \mu_t h_t(\vpi)\mright) \le -\alpha \sum_{t \in [t_1,t_2]} \mleft( \lambda_t + \mu_t \mright)\le -\alpha \sum_{t \in [t_1,t_2]}  \mu_t$, which implies that $\vpi$ is $\delta$-safe for any $\delta\in[0,1]$. 
\end{proof}

\section{Proofs for \Cref{sec:guarantees}}

\lmEvent*

\begin{proof}
    By Lemmas~\ref{lm:safeAdv} and~\ref{lm:safeStoc}, we have that in both settings there exists a safe policy with probability at least $1-\delta$. Moreover, by \Cref{eq:interval regret} with probability at least $1-\delta$ the regret of the primal is upperbounded by $M_{\cI} \cump$ for each interval $\cI=[t_1,t_2]$,   $t_1,t_2 \in [T]$. Applying a union bound sufficies to show that the two events hold simultaneously with probability at least $1-2\delta$.
    Then, the statement directly follows from \cref{lm:boundLagrangian}.
\end{proof}

\hbound*
\begin{proof}
    By the definition of event $\CE$ we have that $\mu_\tau\le 2/\alpha$.
    Moreover, by \Cref{lm:ogd} it holds that $\mu_{\tau} \ge \etar\sum_{t \in [\tau-1]} h_t(\vpi_t)$.
    Hence, $\sum_{t \in [\tau]} h_t(\vpi_t) \le \mu_{\tau}/\etar +1 \le 2/(\etar\alpha)+1$.
\end{proof}

\lemmaBeforeEnd*

\begin{proof}
    We consider two cases.
    \begin{OneLiners}
        \item If $\tau =T$, then  
        \[
            \sum_{t \in [\tau]} \lambda_t g_t(x_t)\ge -\cumdb \ge T-\tau-\frac{1}{\rho}-\cumdb.
        \]
        \item Otherwise, if $\tau<T$,
        \begin{align*}
            \sum_{t \in [\tau]} \lambda_t g_t(x_t) & \ge\frac{1}{\rho}\sum_{t \in [\tau]} g_t(x_t)-\cumdb[\tau] 
             = \frac{1}{\rho} \sum_{t \in [\tau]}  \mleft( c_t(x_t) -\rho \mright) - \cumdb[\tau]\\
            & = \frac{1}{\rho} \mleft(B-1 -\tau\rho\mright) - \cumdb[\tau]\\
            & = \mleft(T-\tau-\frac{1}{\rho}\mright) - \cumdb[\tau].
        \end{align*}
        where the first inequality follows by the no-regret guarantee of the dual regret minimizer with respect to the fixed choice of $\lambda= 1/\rho$, and then we use the definition of $g_t$ and the fact that $\tau$ is the time at which the budget is depleted, that is the round in which the available budget becomes strictly smaller than 1 (see \cref{alg:meta alg}).
    \end{OneLiners}
    This concludes the proof.
\end{proof}

\lemmaRegretUnif*

\begin{proof}
    Let $\vpi^*$ be a $(\delta,q,\OPT)$-optimal policy. 
    Then, we have that  
    \begin{align*}
    \sum_{t \in [\tau]} f_t(x_t) &\ge \sum_{t \in [\tau]} \mleft( f_t(\vpi^*)-\lambda_t g_t(\vpi^*)-\mu_t h_t(\vpi^*) +\lambda_t g_t(x_t)+\mu_t h_t(x_t) \mright) - \mleft(\frac{3}{\alpha}+1\mright) \cump[\tau,\delta]  \\
    &\ge\sum_{t \in [\tau]} \left( f_t(\vpi^*)+\lambda_t g_t(x_t)+\mu_t h_t(x_t) \right) - \frac{3}{\alpha}\cumg[\tau,\delta] -\mleft(\frac{3}{\alpha}+1\mright) \cump[\tau,\delta]\\
    &\ge\sum_{t \in [\tau]} f_t(\vpi^*)+   \sum_{t \in [\tau]} \lambda_t g_t(x_t) - \frac{3}{\alpha}\cumg[\tau,\delta] -\mleft(\frac{3}{\alpha}+1\mright) \cump[\tau,\delta] -\cumdr[\tau]\\
    & \ge \sum_{t \in [\tau]} f_t(\vpi^*)+ T-\tau-\frac{1}{\rho} - \frac{3}{\alpha}\cumg[\tau,\delta] -\mleft(\frac{3}{\alpha}+1\mright) \cump[\tau,\delta] -\cumdr[\tau]-\cumdb\\
    & \ge \sum_{t \in [T]} f_t(\vpi^*)  -\frac{1}{\rho} - \frac{3}{\alpha}\cumg[\tau,\delta] -\mleft(\frac{3}{\alpha}+1\mright) \cump[\tau,\delta] -\cumdr[\tau]-\cumdb \\
    & \ge q T\OPT - \frac{1}{\rho}  -\mleft(\frac{3}{\alpha}+1\mright) \mleft(\cump[T,\delta]+\cumg\mright) -\cumdr-\cumdb,
    \end{align*}
    where the first inequality comes from the regret bound of the primal regret minimizer, the second follows by the definition of $(\delta,q,\OPT)$-optimal policy, the third follows by the no-regret guarantee of the dual regret minimizer with respect to action $\mu=0$, the fourth one follows from Lemma~\ref{lm:highG}. Finally, the fifth inequality follows from the fact that $f_t(\cdot)\in [0,1]$, and the last one is by definition of $(\delta,q,\OPT)$-optimal policy. This proves our statement.
\end{proof}

\LemmaOptStoc*

\begin{proof}
    Let $\vpi^*$ be an optimal solution to $\LP_\gamma$.
    We show that, with probability at least $1-\delta$, the policy $\vpi^*$ is $(\delta,1,\OPT_\gamma)$-optimal, proving the statement.

    First, by Azuma–Hoeffding inequality we have that, for $t'\in [T]$, with probability at least $1-\delta$
    \[
        \sum_{t \in [t']} \mleft(\lambda_t g_t(\vpi^*) + \mu_t h_t(\vpi^*) - \lambda_t \E_\distr g(\vpi^*)-\mu_t \E_\distr h(\vpi^*)\mright)\le \mleft(\mu_{[t']}+\frac{1}{\alpha}\mright)\sqrt{2T\log\mleft(\frac{1}{\delta}\mright)},
    \]
    where $\mu_{[t']}$ is the largest dual multiplier $\mu_t$ observed up to $t'$. Notice that we cannot upper bound it righ away as $2/\alpha$ because here we are not requiring event $\CE$ (\Cref{def:event}) to hold. Then, assuming $T>2$, by taking a union bound over all possible rounds $t'$, we get that the following inequality holds with probability at least $1-\delta$ simultaneously for all $t'\in [T]$,
    \begin{multline*}
        \sum_{t \in [t']} \mleft(\lambda_t g_t(\vpi^*) + \mu_t h_t(\vpi^*) - \lambda_t \E_\distr g(\vpi^*)-\mu_t \E_\distr h(\vpi^*)\mright)\le \mleft(\mu_{[t']}+\frac{1}{\alpha}\mright)\sqrt{2T\log\mleft(\frac{T}{\delta}\mright)}\le \mleft(\mu_{[t']}+\frac{1}{\alpha}\mright) \cumg.
    \end{multline*}

    Similarly, we can prove that 
    \[ \mleft| \sum_{t \in [T]} \mleft(f_t(\vpi^*) -  \E_\distr f(\vpi^*)\mright)\mright| \le 2 \sqrt{T \log\mleft(\frac{2T}{\delta}\mright)}=\cumg \]
    holds with probability at least $1-\delta$. Then,
    \[
        \sum_{t \in [T]} f_t(\vpi^*)\ge \sum_{t \in [T]}\E_\distr f(\vpi^*) -\cume= \OPT_\gamma - \cume[T,\delta][].
    \]
    Assuming $T>2$ and applying an union bound, the statement follows.
\end{proof}

\LemmaOptAdv*

\begin{proof}
    Let $\safep$ be a strictly feasible policy such that $\alpha=-\max_{t \in [T]}\max\mleft\{  g_t(\safep), h_t(\safep) \mright\}$, with $\alpha>0$,  and let $\vpi^*\in \argmax_{\vpi\in\Pi} \sum_{t \in [T]} f_t(\vpi)$ be an optimal unconstrained policy.
    It holds $\sum_{t \in [T]} f_t(\vpi^*)\ge T \OPT_{\bar \gamma}$ since the optimal uncostrained policy is better than the optimal constrained policy, which is a solution to $\LP_{\bar \gamma}$.

    Then, consider the policy $\hat\vpi$ such that, for each $v\in\cV,b\in\cB$,
    \[
    \hat\vpi(v)_b= \frac{1}{1+\alpha} \safep(v)_b + \frac{\alpha}{1+\alpha}\vpi^*(v)_b,
    \]
    where, given a policy $\vpi$, we denote by $\vpi(v)_b$ the probability of bidding $b$ under valuation $v$. 
    
    At each iteration we have that both the budget and the ROI constraints are satisfied by the policy $\hat\vpi$ (in expectation with respect to $\hat\vpi$).
    Indeed, for each $t \in [T]$, we have thath $g_t(\hat \vpi)= \frac{1}{1+\alpha} g_t( \safep)+\frac{\alpha}{1+\alpha} g_t( \vpi^\ast)\le \frac{-\alpha}{1+\alpha}+\frac{\alpha}{1+\alpha}\le 0$.   Similarly, we can prove that for each $t \in [T]$ it holds $h_t(\hat \vpi)\le 0$.
    Then, the policy $\hat\vpi$ satisfies the condition $\sum_{t \in [t']} ( \lambda_t g_t(\hat\vpi) + \mu_t h_t(\hat\vpi))\le 0$ for each $t' \in [T]$.
    Moreover,
    \begin{align*}
    \sum_{t \in [T]} f_t(\hat \vpi)&=\sum_{t \in [T]} \mleft(\frac{1}{1+\alpha}f_t(\safep)+\frac{\alpha}{1+\alpha}f_t( \vpi^*) \mright) \\&\ge \sum_{t \in [T]} \frac{\alpha}{1+\alpha}f_t( \vpi^*) \\& \ge  \frac{\alpha}{1+\alpha} \OPT_{\bar\gamma},
    \end{align*}
    which satisfies the first condition of \Cref{def: good policy}.
    This concludes the proof.
\end{proof}

\section{Proofs for \Cref{sec:generalization}}
\newAssumption* 
\begin{proof}
    First, we need to show that there exists a $\delta$-safe policy. In particular, we show that there exists a policy $\safep$ such that, for any time interval $\cI=[t_1,t_2]$, it holds 
    \begin{equation}\label{eq:safe modified}
        \sum_{t \in \cI}\mleft( \lambda_t g_t(\safep)  +\mu_t h_t(\safep)\mright)\le \cumg-\alpha \sum_{t \in \cI} (\mu_t+\lambda_t).
    \end{equation}
    To do that, we show that the interval $\cI$ can be split in smaller intervals of length $k$, and for each of such smaller intervals $\cI'$, it holds
    \begin{equation*}\label{eq:small interval}
        \sum_{t \in \cI'}\mleft( \lambda_t g_t(\safep)  +\mu_t h_t(\safep)\mright)\le 2 k^2 \eta_B- \alpha \sum_{t \in \cI'}(\mu_t+\lambda_t).
    \end{equation*}
    We show that this holds for any $\cI'$ of length $k$ in \cref{lm:aux}.
    Then, the cumulative sum on the original interval $\cI$ is at most 
    \begin{align*}
    \sum_{t \in \cI}\mleft( \lambda_t g_t(\safep)  +\mu_t h_t(\safep)\mright)&\le \mleft\lceil\frac{|\cI|}{k}\mright\rceil\mleft(2 k^2 \eta_B- \alpha \sum_{t \in \hat\cI}(\mu_t+\lambda_t)\mright)\\&\le 2Tk\etab -\alpha\sum_{t\in\cI}(\mu_t+\lambda_t)\le \cumg-\alpha\sum_{t\in\cI}(\mu_t+\lambda_t),
    \end{align*}
    where we set $\hat\cI\in\argmax_{\cI':|\cI'|=k} \sum_{t\in\cI'}(\mu_t+\lambda_t)$. This shows that \cref{eq:safe modified} holds for any interval $\cI$.

    Then, we can show that a $\mleft(\delta,\alpha/(1+\alpha),\OPT_{\bar\gamma}\mright)$-optimal policy exists. In particular, by defining a policy $\hat\vpi$ as in the proof of \cref{lm:optimal policy adv}, we have
    \begin{align*}
        \sum_{t \in \cI} ( \lambda_t g_t(\hat\vpi) + \mu_t h_t(\hat\vpi))  &= \frac{1}{1+\alpha}\mleft(\sum_{t \in \cI} ( \lambda_t g_t(\safep) + \mu_t h_t(\safep))\mright)+\frac{\alpha}{1+\alpha}\mleft(\sum_{t \in [t']} ( \lambda_t g_t(\vpi^*) + \mu_t h_t(\vpi^*))\mright) \\
        & \le \cumg- \frac{\alpha}{1+\alpha} \sum_{t \in \cI} (\mu_t+\lambda_t)    +  \frac{\alpha}{1+\alpha} \sum_{t \in \cI} ( \lambda_t  + \mu_t)\\
        &\le \cumg \le \frac{3}{\alpha} \cumg,
    \end{align*}
    where the first inequality is by \cref{eq:safe modified}. The first condition of \cref{def: good policy} can be shown to hold with the same steps of \cref{lm:optimal policy adv}. This concludes the proof. 
\end{proof}

\begin{lemma}\label{lm:aux}
    For any time interval $\cI$ of length $k$, there exist a policy $\safep\in\Pi$ for which it holds
    \begin{equation*}
        \sum_{t \in \cI}\mleft( \lambda_t g_t(\safep)  +\mu_t h_t(\safep)\mright)\le 2 k^2 \eta_B- \alpha \sum_{t \in \cI}(\mu_t+\lambda_t).
    \end{equation*}
\end{lemma}
\begin{proof}
    Given an interval $\cI$ of lenght $k$, let $(\lambdamax, \mumax)$ be the largest Lagrangian multipliers in the interval $\cI$, and let $(\lambdamin, \mumin)$ be the smallest Lagrangian multipliers in such interval.
    Let $G\defeq \max\mleft\{\lambdamax-\lambdamin, \mumax-\mumin\mright\}$. Then, we have $G\le k \etab$ since $\etab$ is more aggressive than $\etar$ (see \cref{sec:dual rm}), and there are at most $k$ gradient updates in the interval. Then,
    \begin{align*}
        \sum_{t \in \cI} \mleft(\lambda_t g_t(\safep)  +\mu_t h_t(\safep)\mright) &\le  \sum_{\substack{t \in \cI:\\g_t(\safep)>0}} \hspace{-.2cm}\lambdamax g_t(\safep) + \hspace{-.2cm}\sum_{\substack{t \in \cI:\\g_t(\safep)\le 0}}\hspace{-.2cm} \lambdamin g_t(\safep) +\hspace{-.2cm}\sum_{\substack{t \in \cI\\:h_t(\safep)>0}}\hspace{-.2cm} \mumax h_t(\safep) +\hspace{-.2cm} \sum_{\substack{t \in \cI:\\h_t(\safep)\le 0}}\hspace{-.2cm} \mumin g_t(\safep) \\
        &\le -\lambdamax \mleft(\sum_{\substack{t\in\cI:\\g_t(\safep)\le 0}} \hspace{-.2cm}g_t(\safep) + \alpha k\mright) + \sum_{\substack{t \in \cI:\\g_t(\safep)\le 0}}\hspace{-.2cm} \lambdamin g_t(\safep)\\
        &\hspace{3cm}-\mumax\mleft(\sum_{\substack{t\in\cI:\\h_t(\safep)\le 0}}\hspace{-.2cm} h_t(\safep) + \alpha k\mright)+\sum_{\substack{t \in \cI:\\h_t(\safep)\le 0}}\hspace{-.2cm} \mumin g_t(\safep)\\
        & \le kG  - \alpha k \lambdamax  + kG - \alpha k \mumax\\
        &\le 2 kG -  \alpha \sum_{t \in \cI} (\lambda_t+\mu_t)\\
        &\le 2 k^2 \etab -  \sum_{t \in \cI} (\lambda_t+\mu_t),
    \end{align*}
    where the second inequality comes from $\sum_{t \in \cI} g_t(\safep)\le -\alpha k$ and $\sum_{t \in\cI} h_t(\safep)\le -\alpha k$. This concludes the proof.
\end{proof}

\section{Applications}\label{app:app}

First, \Cref{app:proofprimal} proves the no-interval regret guarantees for \Cref{alg:primal regret minimizer}, thereby demonstrating that our framework can be applied for bidding in repeated non-truthful auctions. Then, \Cref{app:furtherapp} describes how to apply our primal-dual template to second-price auctions with continuous valuations and bids. 

\subsection{Non-Truthful Auctions: Details of \Cref{sec:app}}

Recently, many advertising platforms have been transitioning from the second-price auction format toward a first-price format \citep{akbarpour2018credible,despotakis2021first,paes2020competitive}.This is the case, for example, for Google’s Ad Manager and AdSense platforms \citep{admanager,wong21firstprice}.
It is not clear what is an appropriate online bidding strategy for a budget- and ROI-constrained bidder participating in a series of non-truthful auctions.
While second-price auctions are a truthful mechanism, meaning that bidders can bid their true value and maximize their utility, this is not the case for first-price auctions. This makes existing results for the second-price setting inapplicable to the non-truthful setting~\citep{balseiro2019learning,feng2022online}.

We show that our framework can be used to manage bidding in repeated non-truthful auctions under budget and ROI constraints. We will focus on the case of repeated first-price auctions, and we will make the simplifying assumption of having a finite set of possible valuations and bids. 
Extending our results to the continuous-bid setting is an interesting open problem, and it would amount to designing a suitable primal regret minimizer to plug into our framework. One option to accomplish this would be to adapt techniques designed for the unconstrained setting by \citet{han2020learning,han2020optimal}.

At each round $t\in [T]$, the bidder observes their valuation $v_t$ extracted from a finite set $\cV\subset [0,1]$ of $n$ possible valuations. The set $\cX\subset [0,1]$ is interpreted as the finite set of $\nbid$ possible bids. Let $\beta_t$ be the highest-competing bid at time $t$. In the value-maximizing utility model for each $t$ we have $f_t(x_t)\defeq v_t \indicator{x_t\ge \beta_t}$ \citep{babaioff2021non,balseiro2021landscape},  and the cost function is $c_t(x_t)\defeq x_t\indicator{x_t\ge \beta_t}$, where the indicator function $\indicator{x_t\ge \beta_t}$ specifies whether the bidder won the auction at time $t$. In general, we can handle any reward of the form $f_t(x_t)\defeq (v_t - \omega x_t)\indicator{x_t\ge \beta_t}$, with $\omega\in [0,1]$.
We extend the definition of policies from \Cref{sec:baselines} to model \emph{randomized bidding policies}. Each policy $\vpi\in\Pi$ is now a mapping $\vpi:\cV\rightarrow \Delta_\cX$. We denote by $\vpi(v)_x$ the probability of selecting $x$ under valuation $v$.

In order to apply \Cref{alg:meta alg} and obtain the guarantees of \Cref{thm:main stoc} and \Cref{thm:main adv} we have to design a suitable primal regret minimizer satisfying \Cref{eq:interval regret}. The following result is a rewriting of \Cref{thm:adaptive regret primal short} providing an explicit regret bound.

\begin{restatable}{theorem}{primalreg}\label{thm:adaptive regret primal}
    Let $\nval=1$, $\eta\defeq 1/\sqrt{\nbid T}$, $\xi\defeq 1/(2\sqrt{\nbids T})$, $\sigma\defeq 1/T$, and assume that $\eta\le 1/M_{[T]}$. For any $\delta>0$, \emph{EXP3-SIX} guarantees that, w.p. at least $1-\delta$, for any interval $\cI$, and for any $x\in\cX$, 
    \[
        \textstyle \sum_{t\in \cI} \mleft( \lossp(x_t) - \lossp(x)\mright)\le M_\cI^2 \cump[T,\delta],\,\,\textnormal{where}
    \]
    \[\cump[T,\delta] \hspace{-0.1cm}\defeq \hspace{-0.1cm} \mleft(\frac{3}{2} \hspace{-0.05cm} + \hspace{-0.05cm}4\log\mleft(\frac{\nbids T}{\delta}\mright)M_\cI^{-1}\hspace{-0.05cm}+\hspace{-0.05cm}\mleft(\log(T)\hspace{-0.06cm}+\hspace{-0.06cm}1\mright)M_\cI^{-2}\mright)\sqrt{mT}.\]
\end{restatable}

\subsection{Proof of \Cref{thm:adaptive regret primal short}}\label{app:proofprimal}

In order to proceed with the analysis of \Cref{alg:primal regret minimizer}, let $\vp_{t+1}\in[0,1]^m$ be the vector of \emph{pre-weights} for time $t+1$, which is defined as
\[
p_{t+1,x}\defeq \frac{\vpi_{t}(v_t)_x\, \textnormal{exp}\{-\eta \tilde \lossp(x)\}}{\sum_{x'\in\cX}\vpi_{t}(v_t)_{x'}\,\textnormal{exp}\{-\eta \tilde \lossp(x')\}}\qquad \textnormal{ for all }x\in\cX.
\]
Then, we have the following intermediate result.

\begin{restatable}{lemma}{lemmaDiff}
    \label{lemma: concentration}
    Let $\eta>0$ be such that $\eta \,\E_{\vpi}\tilde\lossp(x)<1$ for all $t\in[T]$ and $\pi\in \Pi$.
    Then, for any $t\in[T]$, and $x'\in\cX$, it holds
    \[
        \E_{\vpi_t(v_t)}\mleft[ \tilde\lossp(x)\mright] - \tilde\lossp(x')\le \frac{1}{\eta}\log\mleft(\frac{p_{t+1,x'}}{\vpi_{t}(v_t)_{x'}}\mright) + \frac{\eta}{2}\E\mleft[\tilde\lossp(x)^2\mright].
    \]
\end{restatable}

\begin{proof}
    Let $\vy\defeq \vpi_t(v_t)\in\Delta_\nbid$.
    By the fact that, for any $n\ge 0$, $e^{-n} \le 1 - n + n^2 / 2$, we have that
    \begin{align*}
        - \log \E_{\vy}\mleft[e^{-\eta\tilde\lossp(x)}\mright]& \ge - \log\mleft(1-\eta\E_{\vy}\mleft[\tilde\lossp(x)\mright]+\frac{\eta^2}{2}\E_{\vy}\mleft[\lossp(x)^2\mright]\mright)\\
        \log \E_{\vy}\mleft[e^{-\eta\tilde\lossp(x)}\mright] &\le -\eta\E_{\vy}\mleft[ \tilde\lossp(x)\mright]+\frac{\eta^2}{2}\E_{\vy}\mleft[\tilde\lossp(x)^2\mright],
    \end{align*}
    where the second inequality holds since, by assumption, $\eta\E_{\vy}\mleft[\tilde\lossp(x)\mright]<1$, which implies that the argument of the logarithm is strictly greater than 0.
    Then, by definition of the preweights $\vp_{t+1}$, we have that, for any $x'\in\cX$,
    \begin{align*}
        \eta\E_{\vy}\mleft[ \tilde\lossp(x)\mright]& \le - \log \E_{\vy}\mleft[e^{-\eta\tilde\lossp(x)}\mright] +\frac{\eta^2}{2}\E_{\vy}\mleft[\tilde\lossp(x)^2\mright]
        \\[2mm] &
       = - \log \mleft(\frac{\pi_{t,x'}e^{-\eta\tilde\lossp(x')}}{p_{t+1,x'}}\mright) +\frac{\eta^2}{2}\E_{\vy}\mleft[\tilde\lossp(x)^2\mright].
    \end{align*}
    This yields
    \[
    \E_{\vpi_t(v_t)}\mleft[ \tilde\lossp(x)\mright] - \tilde\lossp(x')\le \frac{1}{\eta}\log\mleft(\frac{p_{t+1,x'}}{\pi_{t,x'}(v_t)}\mright) + \frac{\eta}{2}\E_{\vpi_t(v_t)}\mleft[\tilde\lossp(x)^2\mright],
    \]
    for any possible alternative bid $x'\in\cX$.
\end{proof}

\primalreg*

\begin{proof}
    In order to increase the readability, we will write $\vpi_t$ in place of $\vpi_t(v)$ since $v\in\cV$ is constant throughout the proof (\ie $n=1$).
    
    By definition of $\tilde\lossp$, we have that for any $x\in\cX$ and $\vpi\in\Delta_\cX$, $\E\tilde\lossp(x)\le\E\mleft[\indicator{x=x_t}\lossp(x)/\vpi_x\mright]=\lossp(x)$. Therefore, since by assumption we have $\eta<1/M_{[T]}$, where $M_{[T]}$ is the maximum range of the loss functions $\lossp$ over the time horizon, the assumption of \Cref{lemma: concentration} holds. 
    Then, for any interval $[t_1,t_2]$, with $t_1,t_2\in[T]$, $t_1<t_2$, by \Cref{lemma: concentration} we have that for any $x'\in \cX$, 
    \[
        \sum_{t \in [t_1,t_2]} \mleft(\E_{\vpi_t} \mleft[\tilde \lossp(x) \mright]- \tilde\lossp(x') \mright) \le \sum_{t \in [t_1,t_2]}\mleft( \frac{1}{\eta}\log\mleft(\frac{p_{t+1,x'}}{\pi_{t,x'}}\mright) + \frac{\eta}{2}\E_{\vpi_t}\mleft[\tilde\lossp(x)^2\mright]\mright).
    \]
    
    Moreover we have that  
    \begin{align*}
        \sum_{t \in [t_1,t_2]}\log\mleft(\frac{p_{t+1,x'}}{\vpi_{t,x'}}\mright) & = \log\mleft(\frac{1}{\vpi_{t_1,x'}}\mright) + \sum_{t\in[t_1+1,t_2]} \log \mleft(\frac{p_{t,x'}}{\vpi_{t,x'}}\mright) + \log\mleft(p_{t_2+1,x'}\mright)
        \\[2mm] & 
        \le \log\mleft(\frac{\nbid}{\sigma}\mright) + \sum_{t\in[t_1+1,t_2]}\log\mleft(\frac{1}{1-\sigma}\mright).
    \end{align*}
    The last inequality holds since, for any $t\in [T]$ and $x\in\cX$, 
    \begin{align*}
        \vpi_{t+1,x} & =\frac{(1-\sigma) w_{t,x} e^{-\eta\tilde \lossp(x)} + \frac{\sigma}{m} \sum_{i\in \cX} w_{t,i} e^{-\eta\tilde \lossp(i)}}{\sum_{i\in \cX} \mleft((1-\sigma)w_{t,i}e^{-\eta\tilde\lossp(i)} + \frac{\sigma}{m} \sum_{j\in \cX} w_{t,j} e^{-\eta\tilde \lossp(j)} \mright)}
        \\[3mm] & 
        = \frac{(1-\sigma) w_{t,x} e^{-\eta\tilde \lossp(x)} + \frac{\sigma}{m} \sum_{i\in \cX} w_{t,i} e^{-\eta\tilde \lossp(i)}}{\sum_{i\in \cX} w_{t,i}e^{-\eta\tilde\lossp(i)}}
        \\[3mm] &
        \ge (1-\sigma)\frac{ w_{t,x} e^{-\eta\tilde \lossp(x)}}{\sum_{i\in\cX}w_{t,i}}\frac{\sum_{i\in\cX}w_{t,i}}{\sum_{i\in \cX} w_{t,i}e^{-\eta\tilde\lossp(i)}}
        \\[3mm] & 
        = (1-\sigma) \frac{\vpi_{t,x}e^{-\eta\tilde \lossp(x)}}{ \E_{\vpi_t}\mleft[e^{-\eta\tilde\lossp(i)}\mright]}
        \\[3mm] & 
        = (1-\sigma) p_{t+1,x},
    \end{align*}
    where we used the definition of $\vpi_t$ and $\vp_t$ as per \Cref{alg:primal regret minimizer}.
    
    Then, 
    \begin{equation}\label{eq:hence}
        \sum_{t \in [t_1,t_2]} \mleft(\E_{\vpi_t}\tilde \lossp(x) - \tilde\lossp(x')\mright)  \le \frac{1}{\eta}\mleft(\log\mleft(\frac{\nbid}{\sigma}\mright) + (t_2-t_1-1) \log \mleft(\frac{1}{1-\sigma}\mright)\mright) + \frac{\eta}{2} \sum_{t \in [t_1,t_2]}\E_{\vpi_t}\mleft[\tilde\lossp(x)^2\mright].
    \end{equation}

    \citet[Lemma 1]{neu2015explore} states that given a fixed non-increasing sequence $(\xi_t)$ with $\xi_t\ge 0$, and by letting $\beta_{t,i}$ be a nonnegative random variable such that $\beta_{t,i}\le 2\xi_t$ for all $t$ and $i\in \cX$, then with probability at least $1-\delta$,
    \[
    \sum_{t\in[T]}\sum_{i\in\cX}\beta_{t,i}\mleft(\tilde\lossp(i) -\lossp(i)\mright)\le \log(1/\delta). 
    \] 
    Then, for any bid $i\in\cX$, by setting
    \[
        \beta_{t,j}=\mleft\{\begin{array}{lll}
            \displaystyle
            2\xi_t\indicator{i=j} & \text{\normalfont if } t\in\cI  \\
            0 & \text{\normalfont otherwise}
        \end{array}\mright.,
    \]
    and by applying a union bound we obtain that, with probability at least $1-\delta$,
    \begin{equation}\label{eq:neu cor}
        \sum_{t\in \cI} \mleft(\tilde \lossp(i) - \lossp(i)\mright)\le \frac{M_\cI \log\mleft(\nbid/\delta\mright)}{2 \xi},
    \end{equation}
    where $M_\cI$ is the maximum range of the loss functions $\lossp$ over time interval $\cI$.

    Moreover, from the definition of $\tilde\lossp$ (see \Cref{alg:primal regret minimizer}), we have that
    \begin{equation}\label{eq:x}
        \sum_{t \in \cI} \E_{x\sim\vpi_t}\tilde \lossp(x)= \sum_{t\in \cI} \mleft(\lossp(x_t) - \sum_{x\in\cX} \xi\tilde\lossp(x)\mright).
    \end{equation}
    Finally, given $t\in\cI$, we observe that 
    \begin{equation}\label{eq:moreover}
        \E_{x\sim\vpi_t} \tilde \lossp(x)^2 = \sum_{x\in\cX} (\vpi_{t,x} \,\tilde \lossp(x))\, \tilde \lossp(x) \le M_\cI  \sum_{x\in \cX}\tilde \lossp(x). 
    \end{equation}
    
    Finally, we conclude by showing that, for any $x\in\cX$, with probability at least $1-\delta$, 
    \begin{align*}
        \sum_{t\in [t_1,t_2]} \mleft( \lossp(x_t) - \lossp(x)\mright) & \le \sum_{t\in [t_1,t_2]}\lossp(x_t) + \frac{M_\cI\log(\nbid/\delta)}{2\xi} - \sum_{t\in [t_1,t_2]}\tilde\lossp(x) \hspace{.5cm} (\text{\normalfont by \cref{eq:neu cor}})
        \\[3mm] &
        = \frac{M_\cI\log(\nbid/\delta)}{2\xi} + 
        \sum_{t\in [t_1,t_2]} \sum_{i\in \cX}\xi \tilde\lossp(i) +
        \sum_{t\in [t_1,t_2]} \mleft( \E_{\vpi_t}\tilde\lossp(j) - \tilde\lossp(x)  \mright) \\& \hspace{7.8cm} (\text{\normalfont by \cref{eq:x}})
        \\[3mm] &
        \le \frac{M_\cI\log(\nbid/\delta)}{2\xi} + 
        \sum_{t\in [t_1,t_2]} \sum_{i\in \cX}\xi \tilde\lossp(i)+ \frac{\eta}{2} M_\cI  \sum_{t \in [t_1,t_2]} \sum_{i\in \cX}\tilde \lossp(i)
        \\[3mm] &
        \hspace{2.7cm} + \frac{1}{\eta}\mleft(\log\mleft(\frac{\nbid}{\sigma}\mright) + (t_2-t_1-1) \log \mleft(\frac{1}{1-\sigma}\mright)\mright)
        \\&\hspace{5cm} (\text{\normalfont by \cref{eq:hence} and \cref{eq:moreover}})
        \\[3mm] &
        \le \frac{M_\cI\log(\nbid/\delta)}{2\xi} + (t_2-t_1)\nbid\xi M_\cI 
        + (t_2-t_1)\frac{\eta M_\cI^2\nbid}{2}
        \\[3mm] &
        \hspace{2.8cm}
        +\frac{1}{\eta}\mleft(\log\mleft(\frac{\nbid}{\sigma}\mright) + (t_2-t_1-1) \log \mleft(\frac{1}{1-\sigma}\mright)\mright) .
    \end{align*}

    By setting $\displaystyle\eta=\frac{1}{\sqrt{\nbid T}}$, $\displaystyle\xi=\frac{1}{2\sqrt{\nbids T}}$, with probability at least $1-\delta$,
    \begin{align*}
        \hspace{-1cm}\sum_{t\in [t_1,t_2]} \mleft( \lossp(x_t) - \lossp(x)\mright) &\le \frac{3}{2}M_\cI^2(t_2-t_1)\sqrt{\frac{\nbid}{T}} + \sqrt{\nbid T}\mleft(M_\cI\log\mleft(\frac{m}{\delta}\mright)+\log\mleft(\frac{m}{\sigma}\mright)+(T-1)\log\mleft(\frac{1}{1-\sigma}\mright)\mright)
        \\&
        \le \frac{3}{2}M_\cI^2(t_2-t_1)\sqrt{\frac{\nbid}{T}} + \sqrt{mT}\mleft(2M_\cI\log\mleft(\frac{m}{\delta}\mright)-\log\mleft(\sigma(1-\sigma)^{T-1}\mright)\mright).
    \end{align*}
    By letting $h(z)\defeq-z\log z-(1-z)\log(1-z)$ be the binary entropy function for $z\in[0,1]$, we have that, for $z\in[0,1]$, $h(z)\le z\log(e/z)$ (see, \eg \citet[Corollary 1]{cesa2012mirror}).
    Then, for $\sigma=1/T$, we have that $-\log\sigma(1-\sigma)^{T-1}\le \log(eT)$.
    This yields 
    \[
        \sum_{t\in [t_1,t_2]} \mleft( \lossp(x_t) - \lossp(x)\mright)  \le   \frac{3}{2}M_\cI^2(t_2-t_1)\sqrt{\frac{\nbid}{T}} + \sqrt{mT}\mleft(2M_\cI\log\mleft(\frac{m}{\delta}\mright)+\log(eT)\mright).
    \]
 
    By taking a union bound over all possible intervals $[t_1,t_2]$ we obtain that, with probability at least $1-\delta$,
    \begin{align*}
        \sum_{t\in [t_1,t_2]} \mleft( \lossp(x_t) - \lossp(x)\mright)& \le \frac{3}{2}M_\cI^2(t_2-t_1)\sqrt{\frac{\nbid}{T}} + \sqrt{\nbid T}\mleft(4M_\cI\log\mleft(\frac{\nbids T}{\delta}\mright)+\log(T)+1\mright)\\[2mm]
        & \le M_\cI^2\mleft(\frac{3}{2} + \frac{4}{M_\cI}\log\mleft(\frac{\nbids T}{\delta}\mright)+\frac{\log(T)+1}{M_\cI^2}\mright)\sqrt{mT},
    \end{align*}
    which proves our statement.
    \end{proof}

\subsection{Further Applications: Repeated Second-Price Auctions}\label{app:furtherapp}
In this section, we show that our framework can also be employed by a budget- and ROI-constrained bidder participating in a sequence of repeated second-price auctions. Our model expands upon that of \citet{balseiro2019learning} and \citet{balseiro2022best}, who considered solely budget constraints (though they also considered some non-stationary input models that we do not consider), as well as \citet{feng2022online}, who considered budget- and ROI-constraints, but only the stochastic input model. 

Let $\cV\subseteq [0,1]$ be the set of possible valuations, and let $\cX\subseteq[0,1]$ be the set of available bids.
Let $\beta_t$ be the highest competing bid at time $t$. In second-price auctions, the utility function of the bidder at time $t$ is $f_t:\cX\ni x\mapsto (v_t - \omega \beta_t)\indicator{x\ge \beta_t}$, and the cost function is $c_t:\cX\ni x\mapsto \beta_t\indicator{x\ge \beta_t}$, where the indicator function $\indicator{x\ge \beta_t}$ specifies whether the bidder won the auction. 
The parameter $\omega\in [0,1]$ represents the bidder's \emph{private capital cost} which normalizes the bidder's accumulated valuation with the overall expenditure. This model includes the traditional quasi-linear set-up (\ie $\omega=1$), as well as the value-maximizing utility model (\ie $\omega=0$) (see, \eg \citet{babaioff2021non,balseiro2021landscape}). To simplify the notation, let's assume $\omega=1$. However, it's important to note that the same results readily apply to any $\omega\in [0,1]$.

At each time $t\in [T]$, the bidder observes valuation $v_t\in\cV$, and subsequently chooses a bid $x_t\in\cX$. Then, the bidder observes the realized reward $f_t(x_t)$ and cost $c_t(x_t)$.

By leveraging the truthfulness property of second-price auctions, we now show that it is possible to substitute the primal regret minimizer with a simple closed-form solution. This solution suitably rescales the valuation $v_t$ using the current values of dual multipliers $\lambda_t$ and $\mu_t$. 
 A similar trick is used e.g. in \citet{balseiro2019learning} for only the Lagrange multiplier from the budget constraint.
We start by considering the baseline for primal regret. Given a sequence of inputs $\vgamma$, we have 
\begin{align*}
    x^\ast\in\argmax_{x\in\cX}\mleft( -\sum_{t=1}^T\lossp(x)\mright) &=\argmax_{x\in\cX} \sum_{t=1}^T\mleft(  f_t(x) +\lambda_t (\rho-c_t(x))+\mu_t (f_t(x)-c_t(x))\mright)
    \\&
    =\argmax_{x\in\cX}\sum_{t=1}^T\indicator{x\ge\beta_t}\underbrace{\mleft(v_t-\beta_t-\lambda_t\beta_t+\mu_tv_t-2\mu_t\beta_t\mright)}_{\circled{A}}.
\end{align*}

At each $t$ there are three possible cases: (i) the bidder wants to win the item at time $t$ if $\circled{A}> 0$; (ii) the bidder has to pass on the item if  $\circled{A}< 0$; (iii) the bidder is indifferent if equality holds. Even though the bidder does not know $\beta_t$ at the time of choosing $x_t$, they can achieve such behavior within a second-price auction framework by bidding, at each $t\in [T]$,
\begin{equation}\label{eq:second}
x_t=\frac{(1+\mu_t)v_t}{1+\lambda_t+2\mu_t}.
\end{equation}
The value obtained by following such bidding policy is such that, for any time interval $\cI$ we have 
\[
\sum_{t\in\cI}^T\mleft[v_t-\beta_t-\lambda_t\beta_t+\mu_tv_t-2\mu_t\beta_t\mright]^+\ge \max_{x\in\cX}\sum_{t\in\cI}\indicator{x\ge\beta_t}\mleft(v_t-\beta_t-\lambda_t\beta_t+\mu_tv_t-2\mu_t\beta_t\mright),
\]
where $[x]^+\defeq\max\{x,0\}$.
Therefore, we have that, with probability 1, the regret accumulated by the primal regret minimizer following \Cref{eq:second} is 0.

\begin{algorithm}[tb]
    \caption{Primal-dual framework for second-price auctions.}
    \label{alg:secondprice}
 \begin{algorithmic}
    \STATE {\bfseries Input:} parameters $B,T$; 
    \STATE {\bfseries Initialization:} $B_1\gets B$;  $\lambda_1=\mu_1=0$
    \FOR{$t = 1, 2, \ldots , T$}
         \STATE{\bfseries{Observe}} valuation $v_t$
         \STATE{\bfseries Primal decision:} $$x\gets \frac{(1+\mu_t)v_t}{1+\lambda_t+2\mu_t}$$
         \STATE Then,
         \[
         x_t\gets \mleft\{\hspace{-1.25mm}\begin{array}{l}
                 \displaystyle
                 x \hspace{.5cm}\text{\normalfont if } B_{t} \ge 1\\ [2mm]
                 \nullx \hspace{.5cm}\text{otherwise}
             \end{array}\mright.
         \]
     \STATE {\bfseries Observe} $f_t(x_t),c_t(x_t)$ and update available resources: $$B_{t+1}\gets B_{t} - c_t(x_t)$$
    \STATE {\bfseries Dual Update:}
    \[\lambda_{t+1} \gets  P_{[0,1/\rho]}\mleft(\lambda_{t}+ \etab \mleft(c_t(b_t)-\rho\mright)\mright)\quad\textnormal{and}\quad
    \mu_{t+1} \gets P_{\R_+}\mleft[\mu_{t}+  \etar \mleft( c_t(b_t) - f_t(b_t) \mright)\mright]\]
    \ENDFOR
 \end{algorithmic}
 \end{algorithm}

 \Cref{alg:secondprice} provides an instantiation of the dual-balancing primal-dual template described in \Cref{alg:meta alg} for the case of second-price auctions. The primal step follows \Cref{eq:second}, while the dual updates follow \Cref{sec:dual rm}. We have that \Cref{thm:main stoc,thm:main adv} hold with 
 \[
    \reg(T,\alpha,\delta)\defeq \nicefrac{1}{\alpha}  +\mleft(\nicefrac{3}{\alpha}+1\mright)\cumg+\cumdr+\cumdb.
 \]
 Therefore, \Cref{alg:secondprice} matches the $\tilde O(\sqrt{T})$ regret upper bound by \citet{feng2022online} in the stochastic setting, while also providing guarantees for in the adversarial setting. This is the first algorithm with provable best-of-both-worlds guarantees for the problem of bidding in repeated second-price auctions under budget and ROI constraints.

\end{document}